\documentclass[11pt]{article}
\pdfoutput=1 
\usepackage[utf8]{inputenc}
\usepackage{amsmath}
\usepackage{dsfont}
\usepackage{comment}
\usepackage{commath}
\usepackage{amssymb}
\usepackage{caption}
\usepackage{amsthm}
\usepackage{mathtools}
\usepackage{multirow}

\usepackage {tikz}
\usepackage[margin=1in]{geometry}
\usepackage{systeme}
\usepackage{colortbl}
\usepackage{nicefrac}
\usepackage{multicol}
\usepackage{tabularx}
\usepackage{makecell}
\usepackage{tcolorbox}
\usepackage[ruled,vlined]{algorithm2e}

\SetCommentSty{mycommfont}
\usepackage{xspace}
\usepackage{subfig}
\usepackage{enumitem}
\usepackage{hyperref}
\hypersetup{colorlinks=true,linkcolor=black,citecolor=[rgb]{0.5,0,0}}
\usepackage[multiple, symbol]{footmisc}
\usepackage[titletoc, title]{appendix}
\newtheorem{theorem}{Theorem}[section]

\newtheorem{corollary}[theorem]{Corollary}
\newtheorem{lemma}[theorem]{Lemma}

\theoremstyle{definition}
\newtheorem{definition}{Definition}[section]
\theoremstyle{definition}

\newtheorem{remark}{Remark}[section]

\DeclareMathOperator*{\argmin}{arg\,min}

\DeclareMathOperator{\poly}{poly}

\newcommand{\R}{\mathbb{R}}

\newcommand{\N}{\mathbb{N}}
\newcommand{\F}{\mathbb{F}}

\renewcommand{\norm}[1]{\left\lVert#1\right\rVert}

\renewcommand{\P}{$\mathsf{P}$}
\newcommand{\NP}{$\mathsf{NP}$}

\newcommand{\kdiam}{$\mathsf{Max}$-$k$-$\mathsf{Diameter}$\xspace}
\newcommand{\tdiam}{$\mathsf{Max}$-$3$-$\mathsf{Diameter}$\xspace}
\newcommand{\kcen}{$k$-$\mathsf{center}$\xspace}
\newcommand{\kmean}{$k$-$\mathsf{means}$\xspace}
\newcommand{\kmed}{$k$-$\mathsf{median}$\xspace}
\newcommand{\minsum}{$k$-$\mathsf{minsum}$\xspace}
\newcommand{\had}{\tilde{\text{Had}}}
\newcommand{\PTAS}{$\mathsf{PTAS}$\xspace}

\newcommand{\np}{$\mathsf{NP}$\xspace}
\newcommand{\p}{$\mathsf{P}$\xspace}

\newcommand{\ptas}{$\mathsf{PTAS}$\xspace}

\newcommand{\comp}{\overline}
\newcommand{\mc}{\mathcal}
\newcommand{\ip}[2]{\langle #1, #2 \rangle}

\renewcommand{\tilde}{\widetilde}

\usepackage{setspace}
\usepackage{todonotes}

\title{\bf Inapproximability of Maximum Diameter Clustering for Few Clusters}
 \author{Henry Fleischmann \\ University of Cambridge
\and Kyrylo Karlov \\ Charles University
\and Karthik C.\ S. \\Rutgers University
 \and 
 Ashwin Padaki \\ Columbia University
 \and Stepan Zharkov \\ Stanford University}

\date{}

\begin{document}
\maketitle

\begin{abstract}
In the \kdiam problem, we are given a set of points in a metric space, and the goal is to partition the input points into $k$ parts such that 
the maximum  pairwise distance between points in the same part of the partition is minimized. \vspace{0.1cm}
 
 The approximability of the \kdiam problem was studied in the eighties, culminating in the work of Feder and Greene [STOC'88], wherein they showed it is \NP-hard to approximate within a factor better than 2 in the $\ell_1$ and $\ell_\infty$ metrics, and \NP-hard to approximate within a factor better than 1.969 in the Euclidean metric. This complements the celebrated 2 factor polynomial time approximation algorithm for the problem in general metrics (Gonzalez [TCS'85]; Hochbaum and Shmoys [JACM'86]). \vspace{0.1cm}

 Over the last couple of decades, there has been increased interest from the algorithmic community to study the approximability of various clustering objectives when the number of clusters is fixed. In this setting, the framework of coresets has yielded \PTAS for most popular clustering objectives, including \kmean, \kmed, \kcen, \minsum, and so on.\vspace{0.1cm}

In this paper, rather surprisingly, we prove that even when $k=3$, the \kdiam problem is \NP-hard to approximate within a factor of $1.5$ in the $\ell_1$-metric (and Hamming metric) and \NP-hard to approximate within a factor of $1.304$ in the Euclidean metric.\vspace{0.1cm}

Our main conceptual contribution is the introduction of a novel framework called \emph{cloud systems} which embed hypergraphs into $\ell_p$-metric spaces such that the chromatic number of the hypergraph is related to the quality of the \kdiam clustering of the embedded pointset. Our main technical contributions are the constructions of nontrivial cloud systems in the Euclidean and $\ell_1$-metrics using extremal geometric structures.    
\end{abstract}

\thispagestyle{empty} 
\clearpage
 \pagenumbering{arabic}

%




\section{Introduction} \label{sec: introduction}

The \kdiam problem is the task of optimizing a  classic clustering objective, where given a set of points in a metric space, we are required to  partition the points into $k$ parts so as to minimize
the maximum pairwise distance between points in the same part of the partition (see Section~\ref{sec: preliminaries} for a formal definition). 
This clustering objective was actively studied in the eighties under the lens of approximation. One of the main advantages of \kdiam over other clustering objectives such as \kcen, \kmean, and \kmed is that \kdiam is not center-based, and thus one does not have to worry about the quality of centers in the applications of this clustering objective.

By the early 1980s, the \kdiam problem was known to be \NP-complete even in the Euclidean plane (implicit in \cite{fowler1981optimal}). Thus, the attention of the community turned towards understanding its approximability in various metric spaces.  There were a series of works all providing 2-approximation to the \kdiam problem in general metrics with improved polynomial runtimes \cite{hochbaum1985best, Hochbaum_Shmoys_1986,Gonzalez_1985,Feder_Greene_1988}. Hochbaum and Shmoys \cite{Hochbaum_Shmoys_1986} showed that this factor cannot be improved (assuming \NP$\neq$\P) for general metrics. Thus, the focus shifted to $\ell_p$-metrics. 

Gonzalez \cite{Gonzalez_1985} showed that the \kdiam problem in the Euclidean metric  is \NP-hard to approximate to a factor of 1.732. This was improved by Feder and Greene \cite{Feder_Greene_1988} who showed that the \kdiam problem in the Euclidean metric  is \NP-hard to approximate to a factor of 1.969. Moreover, they showed that the \kdiam problem is \NP-hard to approximate within a factor better than 2 in both the $\ell_1$ and $\ell_{\infty}$ metrics. More importantly, all the results of Feder and Greene hold when the input is in two dimensions. 

At first glance, it looks like we have hit a road block in understanding the approximability of the \kdiam problem. On the one hand, we have completely understood its approximability in the general metrics, $\ell_1$, and $\ell_\infty$ metrics (even in the plane!). On the other hand, the small gap in our knowledge about the Euclidean \kdiam problem is  much like other Euclidean clustering problems (such as Euclidean \kcen \cite{Feder_Greene_1988,mentzer1988approximability}), which have been open for decades with no real tools of attack to bridge the gap in approximability.

Over the last two-and-a-half decades, a lot of effort has been invested by the algorithmic community to understand the approximability of popular clustering objectives when the number of clusters is \emph{fixed} (in $\ell_p$-metrics). In fact, a powerful paradigm has emerged to provide polynomial time approximation schemes (\PTAS)  for all these problems: \emph{Coresets} \cite{feldman2020introduction,cohen2021new}. In the literature, when $k$ is fixed, a \PTAS is already known for  \kcen~\cite{badoiu2002approximate}, \kmean~\cite{kumar2010linear,chen2009coresets}, \kmed~\cite{kumar2010linear, chen2009coresets}, \minsum~\cite{de2003approximation}, and other popular clustering objectives. 

\begin{center}\emph{
   Does \kdiam in $\ell_p$-metrics admit a \PTAS when $k$ is fixed?}
\end{center}

An argument by Megiddo \cite{Megiddo_1990} proves that \tdiam is \NP-hard to approximate to a factor better than 2 in the $\ell_\infty$-metric by a reduction from the 3-coloring problem. Thus, we explore the question in the Euclidean metric and the $\ell_1$-metric. 

\begin{center}\emph{
    What is the complexity/approximability of \kdiam\\ in the $\ell_1$-metric and the $\ell_2$-metric when $k$ is fixed?}
\end{center}

\subsection{Our Results}

The main contributions of this paper are strong hardness of approximation results for \kdiam in the Euclidean metric and the $\ell_1$-metric, even when $k=3$.  
Our results are surprising, as it is unclear \textit{a priori} why \tdiam does not admit a \ptas much like the aforementioned other clustering objectives. Note that it has been observed in the literature (for example \cite{Megiddo_1990}) that \kdiam can be reduced to the $k$-coloring problem; since $2$-coloring is in \P, we have that \kdiam is in \P\ when $k=2$. 

First, we present our result for the $\ell_1$-metric.

\begin{theorem}[formalized in Theorem \ref{thm: L1 3/2}]
    For every $\varepsilon>0$ and $k\geq 3$, approximating \kdiam in the $\ell_1$-metric (and the Hamming metric) to a factor $1.5-\varepsilon$ is \NP-hard.
\end{theorem}

To the best of our knowledge, there is no known polynomial time algorithm that achieves a factor less than 2 for even \tdiam in the $\ell_1$-metric. Actually, we prove the above \NP-hardness for the \tdiam problem, and the hardness of approximation continues to hold for the \kdiam problem  when $k>3$, simply because we can reduce a hard instance of the \tdiam problem  to a hard instance of the \kdiam problem by adding $k-3$ many outliers to the input instance. 

Additionally, since the above result also holds in the Hamming metric, it implies that the \tdiam problem in the Euclidean metric is \NP-hard to approximate to a factor better than $\sqrt{1.5}\approx 1.224$, using the exact same construction. Interestingly, there is a better than 2-approximation algorithm for the \tdiam in the $\ell_2$-metric. As a straightforward consequence of the \ptas in \cite{badoiu2002approximate} for the $k$-center problem, there exists an efficient $(\sqrt{2}+\varepsilon)$-approximation algorithm for the \kdiam problem in the Euclidean metric (for any $\varepsilon>0$ and constant $k\in\N$).

That said, we are able to exploit the structure of the Euclidean metric to improve on the $\sqrt{1.5}$-factor \NP-hardness and prove the following: 

\begin{theorem}[formalized in Theorem \ref{thm: L2 1.304}]\label{thm:introell2}
 For every   $k\geq 3$, approximating \kdiam in the $\ell_2$-metric to a factor $1.304$ is \NP-hard.
\end{theorem}

This result is particularly noteworthy, since the Euclidean metric is near-isometrically embeddable into all $\ell_p$-metrics, and so hardness within a factor of 1.304 extends to all the $\ell_p$-metrics.

The result in Theorem~\ref{thm:introell2}  is surprising on two counts. First, in the literature we do not know how to exploit Euclidean metric structure to obtain higher hardness of approximation factors for sum clustering objectives (such as \kmean, \kmed, \minsum) better than the ones that can be derived from the Hamming metric hard instances~\cite{CK19,CKL22}. Even in the max clustering objectives (such as \kcen and \kdiam), \emph{planar geometry} is exploited to obtain strong hardness of approximation results~\cite{Feder_Greene_1988,mentzer1988approximability}, and Theorem~\ref{thm:introell2}, to the best of our knowledge, is the first time where high-dimensional Euclidean geometry is utilized to present strong hardness of approximation results.\footnote{Even in inapproximability results of other geometric optimization problems such as Euclidean Travelling Salesman Problem \cite{trevisan2000hamming} or (discrete) Euclidean Steiner Tree problem \cite{Fleischmann2023}, the structure of high-dimensional Euclidean spaces hinders (more than helps) in proving anything better than what can be inferred from the Hamming metric for those problems.} Second, whenever there are nontrivial approximation algorithms for a computational problem in a specific metric space, it is  notoriously hard to improve hardness of approximation results in that metric space for that problem.

Lastly, it is worth noting that our hardness results continue to hold even when the input pointset (of size $n$) is restricted to $O(\log n)$ dimensions. In the Euclidean metric, this follows from applying the Johnson-Lindenstrauss lemma \cite{JL84}; in the Hamming metric, it follows from our proof (where we can replace the  Hadamard codes used in the proof in this paper by small-biased sets~\cite{naor1990small}, consequently bringing the dimension down to   $\Theta(\log n)$). Finally, we remark  that when the number of dimensions is bounded by a constant, the \tdiam problem is in \P.

The state-of-the-art hardness bounds and approximation algorithms for the \kdiam problem, for constant $k\geq 3$, is summarized in Table~\ref{table}.
 
\begin{table}[h]
\begin{center}\renewcommand{\arraystretch}{1.45}
\begin{tabular}{|c|c|c|}
    \hline
    Metric & \thead{\NP-Hardness\\ Approximation Factor} & \thead{Polynomial Time\\  Approximation Factor} \\\hline 
        
  \multirow{2}{*}{$\ell_\infty$}  & $\mathbf{2-\varepsilon}$ & $\mathbf{2}$ \\ 
    & \cite{Megiddo_1990} &   \cite{Gonzalez_1985} \\ \hline
     \multirow{2}{*}{$\ell_0/\ell_1$}   & $\mathbf{1.5-\varepsilon}$ & $\mathbf{2}$  \\ 
      & {\color{red} [This Paper]} & \cite{Gonzalez_1985} \\\hline
     \multirow{2}{*}{$\ell_2$} & $\mathbf{1.304}$ & $\mathbf{1.415}$  \\ 
      & {\color{red} [This Paper]} & \cite{badoiu2002approximate} \\\hline
\end{tabular}
\caption{State-of-the-art Approximability Results for \tdiam}\label{table}
\end{center}
\end{table}
\subsection{Our Techniques}









\paragraph{Clustering and coloring.} There is a natural connection between the geometric problem of \kdiam and the graph theoretic problem of $k$-coloring. Namely, there is a simple reduction from \kdiam to $k$-coloring where given a pointset $P$, we can construct a graph $G_P$ where points in $P$ correspond to vertices in $G_P$ and edges are drawn between vertices if the distance between their corresponding points exceeds a threshold $\beta>0$. Then, $G_P$ is $k$-colorable if and only if $P$ can be partitioned into $k$ clusters, each of which has diameter at most $\beta$.

\paragraph{A simple approach.} The above reduction from \kdiam to $k$-coloring also motivates a simple approach for reducing $k$-coloring to the approximation version of \kdiam. Let $G$ be a graph, and suppose we can embed the vertices of $G$ as a set of points $P_G$ in a metric space, such that the following property holds: if two vertices are adjacent in $G$, the distance between their corresponding points in $P_G$ is at least $r\beta$ (for some $r > 1$), and if the two vertices are nonadjacent, the distance between their corresponding points in $P_G$ is at most $\beta$. Then, if $G$ is $k$-colorable, there is a clustering of $P_G$ with diameter at most $\beta$. However, if $G$ is not $k$-colorable, then any $k$-clustering of $P_G$  has diameter at least $r\beta$ (i.e., any partition of $P_G$ into $k$ clusters contains at least one cluster of diameter at least $r\beta$). Therefore, approximating \kdiam within a factor strictly less than $r$ is at least as hard as determining if $G$ is $k$-colorable.

We formalize this technique as follows. For a metric space $(X, \textsf{dist})$, a graph $G = (V,E)$, and an approximation ratio $r>1$, we say that an \textit{$r$-embedding} of $G$ is a map $\varphi:V \to X$ such that, for some $\beta>0$,
\begin{enumerate}
    \item if $(u,v) \notin E$ then $\textsf{dist}(\varphi(u), \varphi(v)) \le \beta$, and 
    \item if $(u,v) \in E$ then $\textsf{dist}(\varphi(u), \varphi(v)) \ge r\beta$.
\end{enumerate}
We refer to $\beta$ as the \textit{short distance} of the $r$-embedding. 
The first property guarantees that if $G$ is $k$-colorable, then the optimal $k$-clustering diameter of $\varphi(V)$ is at most $\beta$. The second property guarantees that if $\varphi(V)$ has a $k$-clustering with diameter less than $r\beta$, then $G$ is $k$-colorable. 

In particular, if $G$ is a family of graphs on which $k$-coloring is \np-hard, and $\varphi$ is an efficiently computable $r$-embedding of $G$, then it is \np-hard to approximate \kdiam within a factor better than $r$.  We now present a proof-of-concept hardness of approximation result using this $r$-embedding framework.

\paragraph{Warm up in the $\ell_{\infty}$-metric.} We describe a simple $2$-embedding of general graphs into the $\ell_\infty$-metric, as given in \cite{Megiddo_1990}.
    Let $G=(V,E)$ be an arbitrary graph with $V = [n]$. Define a mapping $\varphi : V \to \R^n$ via $\varphi(u) = (a_{u,1},\ldots,a_{u,n})\in\R^n$ given by: $$a_{u,v} = \begin{cases}
        0 & \text{if } (u,v)\in E\\
        2 & \text{if } u=v\\
        1 & \text{otherwise}
    \end{cases}.$$

    It can be checked that $\varphi$ is a $2$-embedding in the $\ell_\infty$-metric. Thus, \kdiam is hard to approximate within a factor of $2$ in the $\ell_\infty$-metric for any $k\geq 3$, proving that the  $2$-approximation algorithm is optimal in this setting.

    We can further extend the above \NP-hardness to $\ell_p$-metrics in the following way. It is known that $3$-coloring remains \np-hard even when restricted to $4$-regular graphs. When $\varphi$ is considered in the $\ell_p$-metric, it is an $r$-embedding of $4$-regular graphs for \[r = \del{\frac{2^{p+1}+6}{10}}^{1/p}.\] Thus, \kdiam (for $k\geq 3$) is hard to approximate within a factor of $\sqrt{7/5}$ in the $\ell_2$-metric. However, in the $\ell_1$-metric, this approach does not give any nontrivial hardness, suggesting a need for different techniques.
    
\paragraph{Trevisan's embedding in the $\ell_1$-metric.}
Trevisan utilized Hadamard codes to prove inapproximability results for geometric problems \cite{trevisan2000hamming}. Recall that for $m$ a power of $2$, a Hadamard code is a set of $m$ elements in $\{0,1\}^m$ whose pairwise $\ell_1$ distances are all $m/2$. We outline a direct application of Trevisan's embedding to obtain a nontrivial hardness of approximation result for \tdiam in the $\ell_1$-metric. In particular, we construct a general $5/4$-embedding of graphs which are $4$-regular and $4$-edge-colorable. It can be shown that $3$-coloring is \np-hard on this family of graphs.

On an input graph $G=(V,E)$, let $m$ be a power of $2$ that is at least $|V|$. Since $G$ is $4$-regular and $4$-edge colorable, we can decompose $E$ as the disjoint union of matchings $M_1,M_2,M_3,M_4\subseteq E$. Then, we embed each $v\in V$ as the concatenation $\varphi(v) = (a_{v,1}, a_{v,2}, a_{v,3}, a_{v,4})\in \{0,1\}^{4m}$, where the strings $a_{v,i}\in\{0,1\}^m$ satisfy the following properties for each $i\in [4]$: \begin{enumerate}
    \item If $(u,v)\in M_i$, then $a_{u,i} = \comp{a_{v,i}}$ and thus $\norm{a_{u,i}-a_{v,i}}_1 = m$.
    \item If $(u,v)\notin M_i$, then $\norm{a_{u,i}-a_{v,i}}_1 = m/2$.
\end{enumerate}
One can construct such an embedding by choosing each $a_{v,i}$ to be either a Hadamard codeword or the complement of a Hadamard codeword. If $(u,v)\in E$, then $\norm{\varphi(u)-\varphi(v)}_1 = \frac{5m}{2}$, and otherwise, $\norm{\varphi(u)-\varphi(v)}_1 = 2m$. Therefore, $\varphi$ is a $5/4$-embedding of $G$, proving that \kdiam is hard to approximate within a factor of $5/4$ in the $\ell_1$-metric for any $k\geq 3$.

As shown by these preliminary results, the $r$-embedding framework is a straightforward way to prove hardness of approximation results for \kdiam. However, it turns out that using this framework is restrictive, as it forces a one-to-one correspondence between vertices and points, and it forbids distances in the range $(\beta,r\beta)$. We can relax both of these restrictions using a different reduction from $k$-coloring to \kdiam. This novel framework, which we call an \textit{$r$-cloud system}, allows us to obtain significantly better inapproximability results in the $\ell_1$-metric and the $\ell_2$-metric. 
    
\paragraph{The $r$-cloud system.} 
We now explain the intuition behind the $r$-cloud system for reductions from coloring to \kdiam. In the previous discussion, we reduced from the problem of $k$-coloring on graphs. For the $r$-cloud construction, a more natural problem to consider is panchromatic $k$-coloring on $k$-uniform hypergraphs.\footnote{In panchromatic $k$-coloring, we are given a $k$-uniform hypergraph, and the task is to color the vertices such that in each hyperedge, every vertex is assigned a distinct color.} A simple padding argument shows that this problem is at least as hard as $k$-coloring on graphs. Next, we explain the definition of $r$-cloud system as a natural framework for reducing this problem to the \kdiam problem. The formal definition and proof of the reduction are given in Section \ref{sec: cloud reduction}.

Let $\mc{H} = (V,E)$ be a $k$-uniform hypergraph and $(X, \textsf{dist})$ be a metric
space. The goal is to produce a pointset $P\subset X$ such that panchromatic $k$-coloring on $\mc{H}$ reduces to approximating \kdiam on $P$ within a factor of $r$. We start by associating each vertex $v\in V$ with a set $P_v\subset X$, where $P := \bigcup_{v\in V} P_v$ and the sets $P_v$ are pairwise disjoint. Note the contrast to the $r$-embedding framework, in which each vertex corresponded to only a single point in $X$. We will also define sets $P_e$ for every edge $e\in E$, which we will motivate later.  In the next couple of paragraphs, we will specify two conditions, \emph{proximity} and \emph{spread}, that we would like $P$ to satisfy in order to achieve the completeness and soundness guarantees for the reduction from panchromatic $k$-coloring to \kdiam. 

For the completeness guarantee of the reduction to hold, we must first show that if $\mc{H}$ is panchromatic $k$-colorable then there is a $k$-clustering of $P$ with low diameter, say at most some threshold $\beta$. Given a $k$-coloring of $\mc{H}$, we create a cluster for each of the $k$ colors, and, for each $v\in V$, we place the entire set $P_v$ in the cluster corresponding to the color of $v$. Since vertices in $\mc{H}$ that are of the same color must be nonadjacent, and we would like every cluster to have diameter at most $\beta$, the proximity condition that $P$ must satisfy is naturally defined as follows:
\begin{tcolorbox}[height=40pt,colback=cyan!10!white]
\begin{equation}
    \text{For every nonadjacent } v,v'\in V\text{ we have } \textsf{diam}(P_v\cup P_{v'})\leq \beta. \tag{proximity}
\end{equation}
\end{tcolorbox}

 Note that $v$ and itself are also considered nonadjacent.  

    Second, for the soundness guarantee of the reduction to hold, we must show that if $\mc{H}$ is not panchromatic $k$-colorable, then every $k$-clustering of $P$ has large diameter---namely, at least $r\beta$. This can be enforced by showing that any $k$-clustering of $P$ with diameter strictly less than $r\beta$ corresponds to a panchromatic $k$-coloring of $\mc{H}$. A natural way of extracting a $k$-coloring of $\mc{H}$ from a $k$-clustering of $P$ is to identify each vertex $v\in V$ with a designated point $\rho(v)\in P_v$. Given a $k$-clustering of $P$, we create a color for each of the $k$ clusters, and we assign each $v\in V$ the color associated with the cluster of $\rho(v)$. Then, provided that the clustering has diameter less than $r\beta$, we have the following condition: for each edge $e = \{v_1,\ldots, v_k\}\in E$, the points $\rho(v_1),\ldots,\rho(v_k)$ must all be in different clusters. To make the condition easier to verify, we localize it to each edge. In particular, we associate the edge $e$ with a set $P_e\subset P$ that contains the points $\rho(v_1),\ldots,\rho(v_k)$. Then, the spread condition that $P$ must satisfy is naturally defined as follows:
\begin{tcolorbox}[height=48pt,colback=cyan!10!white]
\[\begin{array}{c}
    \text{For every $k$-clustering of the set $P_e$ with diameter strictly less than $r\beta$,} \\
    \text{the points $\rho(v_1),\ldots,\rho(v_k)$ must all be in different clusters.} \tag{spread}
\end{array}\]
\end{tcolorbox}

We say $P$ is an $r$-cloud system of $\mathcal{H}$ if it satisfies the above defined proximity and spread conditions. Importantly, neither condition places a restriction on individual distances within the pointset. Indeed, the final pointset $P$ may contain distances in the range $(\beta, r\beta)$, allowing for more general constructions than the $r$-embedding framework.

For convenience, we further constrain the sets $P_e$ without losing reasonable generality. First, it is most natural for $P_e$ to be disjoint with all $P_v$ for $v \not \in e$. To justify this, recall that whether or not an edge $e$ is properly colored depends only on the colors of its constituent vertices $v\in e$. Equivalently, it is reasonable to assume $P_e \subseteq \bigcup_{v \in e} P_v$. Now, let $\tau(e,v) := P_e\cap P_v$. It follows that $P_e$ can be written as \[P_e = \bigcup_{v\in e} P_e\cap P_v = \bigcup_{v\in e} \tau(e,v).\]
Similarly, it is reasonable to assume that $P_v\subseteq \bigcup_{e\in E,\, v\in e} P_e$. Thus, $P_v$ can be written as \[P_v = \bigcup_{\substack{e\in E \\ v\in e}} P_e \cap P_v = \bigcup_{\substack{e\in E \\ v\in e}} \tau(e,v).\] Then, the entire $r$-cloud construction is characterized by the sets $\tau(e,v)$ over all $e\in E$ and $v\in e$. In other words, an $r$-cloud system is simply a pair of mappings $(\rho, \tau)$ satisfying the proximity and spread conditions. This is the framing used in Definition \ref{def: cloud}.

Since $r$-cloud systems can be used to reduce panchromatic $k$-coloring to approximate \kdiam, then by noting that the panchromatic $k$-coloring problem on $k$-uniform hypergraphs is \np-hard for $k\geq 3$, we have the following result:

\begin{theorem}[formalized in Theorem \ref{thm: main reduction}] For $k\geq 3$, suppose there is an efficiently computable $r$-cloud system in a metric space $(X,\textup{\textsf{dist}})$ for the family of $k$-uniform hypergraphs. Then, for any $\varepsilon > 0$, approximating \kdiam in $(X,\textup{\textsf{dist}})$ to a factor $r-\varepsilon$ is \np-hard.
\end{theorem}

In Sections \ref{sec: L1} and \ref{sec: L2}, we construct a $3/2$-cloud system in the $\ell_1$-metric and a $1.304$-cloud system in the $\ell_2$-metric, both for $3$-uniform hypergraphs. This proves that for $k\geq 3$, \kdiam is hard to approximate within a factor of $3/2$ in the $\ell_1$-metric and $1.304$ in the $\ell_2$-metric.

\paragraph{Constructing a $3/2$-cloud system in the $\ell_1$-metric.} Let $\mc{H} = (V,E)$ be a $3$-uniform hypergraph, and for simplicity let $m:=|V|$ be a power of $2$. We associate each vertex $v\in V$ with a Hadamard codeword $\textbf{h}_v\in\{0,1\}^m$, and we define $\rho(v)=(\textbf{h}_v,\textbf{h}_v)\in\{0,1\}^{2m}$. For each edge $e = \{x,y,z\}$, the set $P_e$ is then comprised of pairs of $\textbf{h}_x, \textbf{h}_y, \textbf{h}_z$ and their complements. Specifically, $\tau(e,x)$ consists of $5$ out of the $9$ pairs that can be formed using elements in $\cbr{\textbf{h}_x,\comp{\textbf{h}}_y,\comp{\textbf{h}}_z}$, and $\tau(e,y),\tau(e,z)$ are constructed symmetrically.

With this construction, it can be shown that the proximity condition holds with $\beta = m$. For any vertex $v\in V$, recall that $P_v$ is the union of $\tau(e,v)$ over all $e\in E$ that contain $v$. This means that the only non-complemented codeword used by any point in $P_v$ is $\textbf{h}_v$, and also that no point in $P_v$ contains the complement $\comp{\textbf{h}}_v$. As a result, no two points in $P_v$ can contain complementary codewords. By the distance property of the Hadamard code, $\textsf{diam}(P_v)\leq m$. In the proof of Lemma \ref{lem: L1 cloud}, we slightly extend this argument to show that also $\textsf{diam}(P_v\cup P_{v'})\leq m$ for nonadjacent $v,v'\in V$, which is exactly the proximity condition.

Proving that our construction satisfies the spread condition is 
more involved. We start by fixing an edge $e=\{x,y,z\}$. Then, we argue through casework that any $3$-clustering of $P_{\{x,y,z\}}$ in which $(\textbf{h}_x,\textbf{h}_x)$ and $(\textbf{h}_y,\textbf{h}_y)$ are clustered together necessarily has diameter at least $3m/2$. Since our construction of $\tau(e,\cdot)$ is symmetric in $(x,y,z)$, the spread property follows.

\paragraph{Constructing a $1.304$-cloud system in the $\ell_2$-metric.}

Let $\mc{H} = (V,E)$ be a $3$-uniform graph with $m := |V|$. We associate each vertex $v\in V$ with a standard basis vector $\textbf{e}_v\in\R^m$, and we define $\rho(v)=\textbf{e}_v$. For each edge $e = \{x,y,z\}$, the set $P_e$ is then constructed on the surface of the unit sphere in the three dimensional subspace of $\R^m$ spanned by $\cbr{\textbf{e}_x, \textbf{e}_y, \textbf{e}_z}$. Specifically, consider the region on this sphere where the $\textbf{e}_x$ component is nonnegative, the $\textbf{e}_y$ and $\textbf{e}_z$ components are nonpositive, and all other components are 0. We define $\tau(e,x)$ to be a finite net of points on this region. The sets $\tau(e,y)$ and $\tau(e,z)$ are formed symmetrically. An illustration of these regions is shown below in Figure \ref{fig: intro L2 diagram}.

\begin{figure}[!h]
    \centering
    
    \begin{tikzpicture}[scale=0.6]

        \shade[ball color = gray!40, opacity = 0.2] (0,0) circle (4);
        \draw (0,0) circle (4);
        \draw (-4,0) arc (180:360:4 and 0.6);
        \draw[dashed] (-4,0) arc (180:360:4 and -0.6);
        \draw (0,4) arc (90:270:0.5 and 4);

        \fill[cyan, opacity = 0.3] (0,4) arc (90:188.75:0.5 and 4) -- (-4,0) -- (0,4) -- cycle;
        \fill[cyan, opacity = 0.3] (-4,0) arc (180:98:4 and -0.6);
        \fill[cyan, opacity = 0.3] (-4,0) arc (180:90:4 and 4);
        \fill[green, opacity=0.3] (0,-4) arc (270:360:4 and 4);
        \fill[green, opacity=0.3] (4,0) arc (180:277:-4 and 0.6) -- (0,-4) -- (4,0) -- cycle;
        \fill[green, opacity=0.3] (0,-4) arc (270:360:-0.5 and 3.5) -- (0,-4) -- cycle;
        \fill[red, opacity=0.2] (4,0) arc (0:90:4 and 4) -- (4,0) -- cycle;
        \draw[dashed] (0,4) arc (90:171.5:-0.5 and 4);
        \fill[red, opacity=0.2] (0,4) arc (90:171.5:-0.5 and 4) arc (90:178.5:-3.5 and 0.6) -- (0,4) -- cycle;

        \filldraw(4,0) node[xshift=0.5cm]{\small$-\textbf{e}_x$} circle(2pt);
        \filldraw(0,4) node[yshift=0.3cm]{\small$-\textbf{e}_y$} circle(2pt);
        \filldraw(-0.49,-0.6) node[xshift=-0.4cm, yshift=-0.3cm]{\small$-\textbf{e}_z$} circle(2pt);
        \filldraw(-4,0) node[xshift=-0.4cm]{\small$\textbf{e}_x$} circle(2pt);
        \filldraw(0,-4) node[yshift=-0.4cm]{\small$\textbf{e}_y$} circle(2pt);
        \filldraw(0.49,0.6) node[xshift=-0.1cm, yshift=-0.3cm]{\small$\textbf{e}_z$} circle(2pt); 
    \end{tikzpicture}
    \caption{Spherical regions used to define $\tau(e,\cdot)$ for $e=\{x,y,z\}$. Specifically, $\tau(e,x)$, $\tau(e,y)$, and $\tau(e,z)$ are finite nets of the blue, green, and red regions, respectively.} 
    \label{fig: intro L2 diagram}
\end{figure}

For this construction, the proximity condition holds with $\beta = \sqrt{2}$, and our proof is similar to the proof for the $\ell_1$-metric. As before, recall that for any vertex $v\in V$, the set $P_v$ is the union of $\tau(e,v)$ over all $e\in E$ that contain $v$. This means that the only coordinate on which a point in $P_v$ may be positive is the $v^{\text{th}}$ coordinate, and also that no point in $P_v$ has a negative $v^{\text{th}}$ component. Since $P_v$ consists of unit vectors, it follows directly that $\textsf{diam}(P_v)\leq \sqrt{2}$. In the proof of Lemma \ref{lem: L2 cloud}, we argue that $\textsf{diam}(P_v\cup P_{v'})\leq \sqrt{2}$ for nonadjacent $v,v'\in V$, thus proving the proximity condition.

To prove the spread condition, we fix an edge $e=\{x,y,z\}$. Then, we show that any $3$-clustering of $P_{\{x,y,z\}}$ in which $\textbf{e}_x$ and $\textbf{e}_y$ are clustered together must have diameter greater than $1.304\cdot\sqrt{2}$. Since the number of $3$-clusterings of $P_e$ is large even for small nets $\tau(e,\cdot)$, we verify this property using a computer search.

We remark that no matter how dense of a net we choose, we cannot improve substantially on the factor $1.304$. This is because it is possible to cluster the set $P_{\{x,y,z\}}$ with diameter at most $\sqrt{2+\sqrt{2}}\approx 1.307\cdot\sqrt{2}$ such that $\textbf{e}_x$ and $\textbf{e}_y$ are in the same cluster. Geometrically, the length $\sqrt{2+\sqrt{2}}$ is the distance between $\textbf{e}_y$ and the midpoint of the arc between $\textbf{e}_x$ and $-\textbf{e}_y$. Our hardness factor of $1.304$ was obtained using a sufficiently dense net for which verifying the spread property was computationally tractable. 

\paragraph{Barriers to proving hardness.} In Section \ref{sec: barriers}, we outline two barriers to showing improved hardness of approximation for \kdiam. Recall that one limitation of the $r$-embedding technique is that it forbids distances strictly between $\beta$ and $r \beta$ (for some $\beta>0$). Our first barrier result, formalized in Corollary \ref{cor: L0 barrier}, shows that it is not possible to prove hardness above a factor of $5/3$ in the $\ell_1$-metric with the $r$-embedding framework, or using any construction of pointsets with such a gap in the set of distances.

Our next barrier says that if a pointset $P$ in the $\ell_2$-metric is not contained within a sphere of diameter $\Delta\cdot \sqrt{2}$, where $\Delta$ is the optimal $3$-clustering diameter of $P$, then there is a polynomial time algorithm to approximate \tdiam on $P$ within a factor of $\sqrt{2}-\varepsilon$. So, in a proof of hardness up to a factor of $\sqrt{2}$ in the $\ell_2$-metric, the hard instance must effectively lie within a sphere of diameter $\Delta\cdot\sqrt{2}$. This has an interesting graph theoretic consequence: if \tdiam is hard to approximate within a factor of 2 in the $\ell_1$-metric or within a factor of $\sqrt{2}$ in the $\ell_2$-metric, then the graphs associated with the hard instance must have unbounded odd girth. This is formalized in Theorems \ref{thm: in a sphere barrier} and \ref{thm: odd girth barrier}.

\subsection{Open Problems}

A couple of open problems immediately stem from this work. 

\begin{itemize}
    \item In this work, we proved strong inapproximability results for \kdiam (where $k$ is fixed) in the Euclidean and $\ell_1$-metrics, but our understanding is far from tight in both these metrics, and it remains an important open problem to bridge the gap between the \NP-hardness and polynomial time approximability factors. In order to do so, some of the barriers that need to be overcome are detailed in Section~\ref{sec: barriers}. However, an intriguing (intermediate) direction that is worth exploring is if we can obtain tight hardness of approximation factors for \kdiam in the Euclidean metric (or $\ell_1$-metric) when $k$ is some large constant. 
    \item We believe that the techniques introduced in this paper, in particular the graph embedding to $\ell_1$-metric and Euclidean metric, might be useful to prove improved hardness of approximation results for other geometric problems  as well. For example, in \cite{CKL21} the authors gap reduce graph coloring problem  to the  continuous \kmean problem in high-dimensional $\ell_\infty$-metric, to obtain strong inapproximability results that are higher than the ones derived from covering problems in \cite{CK19,CKL22}. We wonder if it is possible to use our embedding schemes to obtain strong hardness of approximation results for the high dimensional continuous  Euclidean \kmean problem, improving upon the current state-of-the-art 1.36 \NP-hardness factor based on the Johnson Coverage Hypothesis \cite{CKL22}. 
\end{itemize}

\subsection{Organization of the Paper}

The organization of the paper is as follows. In Section \ref{sec: preliminaries}, we overview some definitions and notations  related to the clustering and coloring problems considered throughout the paper. In Section \ref{sec: cloud reduction}, we present the central reduction of the paper using the $r$-cloud system. In Sections \ref{sec: L1} and \ref{sec: L2}, we present specific constructions of $r$-cloud systems in the $\ell_1$-metric and the $\ell_2$-metric, respectively. In Section \ref{sec: barriers}, we outline barriers to proving better hardness of approximation.  

\section{Preliminaries} \label{sec: preliminaries}
This section contains definitions, remarks, and notation that will be used throughout this paper.

\subsection{Diameter Objective}

\begin{definition}[diameter]
    Given a metric space $(X,\textsf{dist})$ and a set $C \subset X$, we say the diameter of $C$ is 
    \[\textsf{diam}(C) := \max_{x,y \in C} \textsf{dist}(x,y).\]
    For a collection of subsets $C_1, \dots, C_k\subset X$, we say 
    \[\textsf{diam}(\{C_1, \dots, C_k\}):=\max_{\substack{i \in [k]\\ x,y \in C_i}} \textsf{dist}(x,y)\]
    is the diameter of the collection.
\end{definition}

\begin{definition}[$k$-clustering]
    Given a set $P$, we say that a collection $\{C_1, \dots C_k\}$ is a $k$-clustering of $P$ if $C_1, \dots C_k$ are pairwise disjoint and $C_1 \cup \cdots \cup C_k = P$.
\end{definition}

\begin{definition}[\kdiam clustering problem]
    Let $(X,\textsf{dist})$ be a metric space and let $k$ be a constant. Given as input a finite set $P \subset X$, find a $k$-clustering $\{C_1, \dots, C_k\}$ that minimizes $\textsf{diam}(\{C_1, \dots, C_k\})$.
\end{definition}

\begin{definition}[$r$-approximate \kdiam clustering problem]
    Let $(X,\textsf{dist})$ be a metric space and let $k$ be a constant. Given as input a finite set $P \subset X$, let
    \[\Delta:= \min \textsf{diam}(\{C_1, \dots, C_k\}),\]
    where the minimum is taken over all possible $k$-clusterings of $P$.
    Find a $k$-clustering $\{C_1, \dots, C_k\}$ with diameter at most $r \Delta$.
\end{definition}

\begin{remark} \label{rem: larger k is harder}
    Let $P$ be a pointset and let $p$ be a point whose distance to $P$ is more than $r \cdot \textsf{diam}(P)$. Then, $r$-approximate \kdiam on $P$ reduces to $r$-approximate $\mathsf{Max}$-$(k+1)$-$\mathsf{Diameter}$\xspace on $P \cup \{p\}$, because any $(k+1)$-clustering whose diameter is at most $r\Delta$ must have $\{p\}$ as a standalone cluster. Thus, if \kdiam clustering is \np-hard, then $\mathsf{Max}$-$(k+1)$-$\mathsf{Diameter}$\xspace is \np-hard. In other words, if \tdiam is hard to approximate to some factor $r$, then so is \kdiam for any $k\geq 3$.
\end{remark}

We now restate the definition of $r$-embedding mentioned in Section \ref{sec: introduction}.

\begin{definition}[$r$-embedding] 
    Let $G=(V,E)$ be a graph and $(X,\textsf{dist})$ a metric space. For $r>1$, we say a map $\varphi : V\to X$ is an $r$-embedding of $G$ with short distance $\beta>0$ if the following two conditions hold for every $u,v\in V$: \begin{itemize}
        \item $(u,v)\notin E\implies\textsf{dist}(\varphi(u), \varphi(v))\leq \beta$
        \item $(u,v)\in E\implies\textsf{dist}(\varphi(u),\varphi(v))\geq r\beta$
    \end{itemize}
\end{definition}

\begin{definition} \label{def: gamma}
    Let $k$ be a constant, $r$ a ratio, and $P$ a pointset. Let $\Delta$ be the optimal $k$-clustering diameter of $P$. Define $\Gamma_{k,r}(P)$ to be the graph in which the points in $P$ are vertices, with edges between any two points whose distance is greater than $r\Delta$. 
\end{definition}

\subsection{Coloring Problems}

\begin{definition}
    A hypergraph $\mc{H}$ consists of a set of vertices $V$ and a set of hyperedges $E$, where each hyperedge $e\in E$ is simply a subset $e\subseteq V$. We say that a hypergraph $\mc{H}$ is \textit{$k$-uniform} if $|e| = k$ for every $e\in E$.
\end{definition}

\noindent We define the following coloring problem on $k$-uniform hypergraphs.

\begin{definition}[panchromatic $k$-coloring problem]
    Let $k$ be a constant. Given a $k$-uniform hypergraph $\mc{H}=(V,E)$ as input, find an assignment of $k$ colors to the vertices of $G$ such that for any hyperedge $e \in E$, all vertices in $e$ are given distinct colors. Such an assignment is called a panchromatic $k$-coloring.
\end{definition}

\begin{remark}\label{rem: panchromatic hard}
    Note that $3$-regular graphs can be viewed as $3$-uniform $2$-regular hypergraphs: edges correspond to vertices, and triples of edges incident to the same vertex correspond to hyperedges. Since $3$-edge coloring is \np-hard on $3$-regular graphs \cite{holyer1981np}, we know that panchromatic $3$-coloring is \np-hard on $3$-uniform hypergraphs.
\end{remark}

\subsection{Hadamard Codes}

For a bit string $x \in \{0,1\}^m$, we denote its bitwise complement by $\overline{x}$.

\begin{definition}[Hadamard code]
    Let $m$ be a power of 2. There exists a subset of $\{0,1\}^m$ of size $m$ such that the Hamming distance between distinct elements is exactly $m/2$. We call this a Hadamard code and denote it by $\text{Had}_m^+$. We also define $\text{Had}_m^-:= \{\comp{\textbf{h}} : \textbf{h} \in \text{Had}_m^+\}$, which can alternatively be seen as an affine shift of $\text{Had}_m^+$ by the vector $\textbf{1} = (1,1,\ldots, 1)$ in the vector space $\F_2^m$. Then, let \[\had_m := \text{Had}_m^+ \cup \text{Had}_m^-.\]
\end{definition}

\noindent Given $\textbf{h}_i,\textbf{h}_j\in \had_m$, we will denote the concatenation of $\textbf{h}_i$ and $\textbf{h}_j$ by $(\textbf{h}_i,\textbf{h}_j)\in \{0,1\}^{2m}$.

\begin{remark}
    The set $\had_m$ has the following property: for every $\textbf{h} \in \had_m$, there is exactly one element in $\had_m$ that is distance $m$ away, namely $\comp{\textbf{h}}$. All other elements have distance $m/2$ to $\textbf{h}$.
\end{remark}

\section{The $r$-cloud System} \label{sec: cloud reduction}

In this section, we introduce a geometric object that can be used to prove inapproximability results for the \kdiam problem. 

\begin{definition}
    For a hypergraph $\mc{H} = (V,E)$, let $I_\mc{H}\subset E\times V$ denote its set of incidences: \[I_\mc{H} = \{(e,v) : e\in E,\, v\in e\}.\]
\end{definition}

\begin{definition}[$r$-cloud system] \label{def: cloud}
    Let $(X,\textsf{dist})$ be a metric space and $U\subseteq X$ be a finite set of points. Let $\mc{H}=(V,E)$ be a $k$-uniform hypergraph. Suppose that $(\rho,\tau)$ is a pair of mappings for which $\rho : V\to U$ is injective, and $\tau : I_\mc{H} \to \mathcal{P}(U)$ satisfies $\rho(v) \in \tau(e,v)$ for all $(e,v)\in I_{\mc{H}}$. For $r>1$, we say that $(U, \rho,\tau)$ is an $r$-cloud system of $\mc{H}$ if the following properties hold for some $\beta > 0$:
    \begin{itemize}
        \item \textbf{(Proximity)} For $v\in V$, define \[P_v := \bigcup_{\substack{e\in E \\ (e,v)\in I_{\mc{H}}}} \tau(e,v).\] For any nonadjacent $v,v'\in V$, we have $\textsf{diam}(P_v\cup P_{v'})\leq \beta$. Note that we consider $v$ to be nonadjacent to itself.
        \item \textbf{(Spread)} For $e\in E$, define \[P_e := \bigcup_{v\in e} \tau(e,v)\] Let $e = \{v_1, \dots, v_k\}\in E$. For any $k$-clustering of $P_e$ with diameter strictly less than $r\beta$, the points $\rho(v_1), \dots, \rho(v_k)\in P_e$ are all in different clusters. 
    \end{itemize}
    We refer to $\beta$ as the \textit{short distance} of $(U,\rho,\tau)$. Furthermore, we say that an $r$-cloud system $(U, \rho,\tau)$ is \textit{efficiently computable} if there are $\poly(|V|)$-time algorithms to identify $U$ and evaluate $\rho$ and $\tau$.
\end{definition}

The existence of an efficiently computable $r$-cloud system gives a polynomial time reduction from panchromatic $k$-coloring to \kdiam.

\begin{theorem}\label{thm: main reduction}
    Let $k\in \N$ be a constant and let $\mc{H} = (V,E)$ be a $k$-uniform hypergraph. Suppose $\mc{H}$ has an efficiently computable $r$-cloud system $(U, \rho, \tau)$ with short distance $\beta$, where $U\subseteq X$ for a metric space $(X,d)$. Let 
    \[P := \bigcup_{(e,v)\in I_\mc{H}}\ \tau(e,v) = \bigcup_{v\in V} P_v = \bigcup_{e\in E} P_e.\]
    Then, there is a polynomial time algorithm which takes as input $\mc{H}$ and outputs $P$ with the following guarantees:
    \begin{itemize}
        \item \textbf{\textup{Completeness:}} If $\mc{H}$ is panchromatic $k$-colorable, then there exists a $k$-clustering of $P$ with diameter at most $\beta$.
        \item \textbf{\textup{Soundness:}} If $\mc{H}$ is not panchromatic $k$-colorable, then any $k$-clustering of $P$ has diameter at least $r\beta$.
    \end{itemize}
\end{theorem}
\begin{proof}
    We note that, since $|I_\mc{H}|\leq n^k$, the pointset $P$ is computable in polynomial time. 
    
    We first prove completeness. Let $\mc{H}$ be panchromatic $k$-colorable, and let $c : V\to [k]$ be a panchromatic $k$-coloring of $\mc{H}$. We define a $k$-clustering $\{C_1,\ldots,C_k\}$ of $P$ given by \[C_i := \bigcup_{v \::\: c(v) = i} P_v.\] To show that this clustering has diameter at most $\beta$, it suffices to show that for any $i\in [k]$ and $x,y\in C_i$, we have $d(x,y)\leq \beta$. This follows from the observation that $x\in P_v$ and $y\in P_{v'}$ for some $v,v'\in V$, and since $c(v) = c(v')$, either $v,v'$ are nonadjacent or $v=v'$. Applying the proximity property of Definition \ref{def: cloud} gives $d(x,y)\leq \textsf{diam}(P_v\cup P_{v'}) \leq \beta$.

    Now, we prove soundness. Suppose $\{C_1, \dots, C_k\}$ is a $k$-clustering of $P$ with diameter less than $r\beta$. We construct a $k$-coloring $c : V \to [k]$ as follows: for each $v\in V$, let $c(v) = i$ if $\rho(v)\in C_i$. Now, let $e=\{v_1,\ldots,v_k\}\in E$ be arbitrary. By the spread property of Definition \ref{def: cloud}, we have that $c(v_1),\ldots,c(v_k)$ are all distinct, and thus $c$ is panchromatic.
\end{proof}

By Remark \ref{rem: panchromatic hard}, the above theorem gives a framework of proving hardness of approximating \kdiam.

\section{Hardness of Approximation in the $\ell_1$-metric} \label{sec: L1}

In this section, we demonstrate the existence of an efficiently computable $3/2$-cloud system of any 3-uniform hypergraph in the $\ell_1$-metric, hence proving that \kdiam is \NP-hard to approximate within a factor of $3/2$ in the $\ell_1$-metric.

\begin{definition} \label{def: L1 cloud}
    Let $\mc{H}$ be a $3$-uniform hypergraph, let $m\geq |V|$ be a power of $2$. Identify $V = [|V|]\subseteq [m]$ so that each vertex $v\in V$ corresponds to a distinct codeword $\textbf{h}_v\in \had_m$. Let \[U =\had_m\times \had_m\subset(\{0,1\}^{2m}, \ell_1)\] We define an injective map $\rho : V\to U$ given by $\rho(v) = (\textbf{h}_v,\textbf{h}_v)$. Moreover, we define $\tau : I_\mc{H}\to \mc{P}(U)$ as follows. For every edge $e = \{x,y,z\}$, choose an arbitrary orientation $\tilde{e} = (x,y,z)$. Then, let $\tau$ be given by: \begin{align*}
        \tau(e,x) &= \cbr{
        \del{{\textbf{h}}_x,{\textbf{h}}_x}, \del{{\textbf{h}}_x,\comp{\textbf{h}}_z}, \del{\comp{\textbf{h}}_y,{\textbf{h}}_x}, \del{\comp{\textbf{h}}_y,\comp{\textbf{h}}_z}, \del{\comp{\textbf{h}}_z,\comp{\textbf{h}}_z}} \\
        \tau(e,y) &= \cbr{\del{\textbf{h}_y,\textbf{h}_y}, \del{\textbf{h}_y,\comp{\textbf{h}}_x}, \del{\comp{\textbf{h}}_z,\textbf{h}_y}, \del{\comp{\textbf{h}}_z,\comp{\textbf{h}}_x}, \del{\comp{\textbf{h}}_x,\comp{\textbf{h}}_x}}\\
        \tau(e,z) &= \cbr{\del{\textbf{h}_z,\textbf{h}_z}, \del{\textbf{h}_z,\comp{\textbf{h}}_y}, \del{\comp{\textbf{h}}_x,\textbf{h}_z}, \del{\comp{\textbf{h}}_x,\comp{\textbf{h}}_y}, \del{\comp{\textbf{h}}_y,\comp{\textbf{h}}_y}}\\
    \end{align*}
    This definition of $\tau$ is illustrated in Figure \ref{fig: L1 tau diagram}.
\end{definition}

\begin{figure}[!h] 
    \centering
    \begin{tikzpicture}
        \filldraw(0,0) node[xshift=-0.6cm,yshift=-0.4cm]{\small$(\textbf{h}_x,\textbf{h}_x)$} circle(2pt);
        \filldraw(2,0) node[xshift=-0.1cm, yshift=-0.4cm]{\small$(\textbf{h}_x,\comp{\textbf{h}}_z)$} circle(1.5pt);
        \filldraw(4,0) node[yshift=-0.4cm] {\small$(\comp{\textbf{h}}_z,\comp{\textbf{h}}_z)$}circle(2pt);
        \filldraw(6,0) node[xshift=0.1cm, yshift=-0.4cm]{\small$(\comp{\textbf{h}}_z,\textbf{h}_y)$}circle(1.5pt);
        \filldraw(8,0) node[xshift=0.7cm,yshift=-0.4cm]{\small$(\textbf{h}_y,\textbf{h}_y)$} circle(2pt);

        \filldraw(1,1.732) node[xshift=-0.8cm,yshift=0.1cm]{\small$(\comp{\textbf{h}}_y,\textbf{h}_x)$} circle(1.5pt);
        \filldraw(3,1.732) node[xshift=0.4cm,yshift=0.4cm]{\footnotesize$(\comp{\textbf{h}}_y,\comp{\textbf{h}}_z)$} circle(1.5pt);
        \filldraw(5,1.732) node[xshift=-0.4cm,yshift=0.4cm] {\footnotesize$(\comp{\textbf{h}}_z,\comp{\textbf{h}}_x)$} circle(1.5pt);
        \filldraw(7,1.732) node[xshift=0.8cm,yshift=0.1cm]{\small$(\textbf{h}_y,\comp{\textbf{h}}_x)$} circle(1.5pt);

        \filldraw(2,3.464) node[xshift=-0.8cm,yshift=0.2cm]{\small$(\comp{\textbf{h}}_y,\comp{\textbf{h}}_y)$} circle(2pt);
        \filldraw(4,3.464) node[yshift=-0.35cm]{\footnotesize$(\comp{\textbf{h}}_x,\comp{\textbf{h}}_y)$}circle(1.5pt);
        \filldraw(6,3.464) node[xshift=0.8cm,yshift=0.2cm]{\small $(\comp{\textbf{h}}_x,\comp{\textbf{h}}_x)$} circle(2pt);

        \filldraw(3,5.196) node[xshift=-0.8cm,yshift=0.3cm]{\small $(\textbf{h}_z,\comp{\textbf{h}}_y)$} circle(1.5pt);
        \filldraw(5,5.196) node[xshift=0.8cm,yshift=0.3cm]{\small $(\comp{\textbf{h}}_x,\textbf{h}_z)$} circle(1.5pt);

        \filldraw(4,6.928)  node[yshift=0.4cm]{\small$(\textbf{h}_z,\textbf{h}_z)$} circle(2pt);

        \draw (0,0) -- (8,0) -- (4,6.928) -- (0,0);
        
        \draw[dashed] (4,0) -- (6,3.464) -- (2,3.464) -- (4,0);

        \filldraw[cyan, opacity=0.3] (0,0) --  node[black, opacity=1, yshift=1.3cm] {\Large $\tau(e,x)$} (4,0) -- (2.25, 3.031) -- (1.75, 3.031) -- (0,0);

        \filldraw[green, opacity=0.3] (4.5,0) -- node[black, opacity=1, xshift=-0.2cm,yshift=1.3cm] {\Large $\tau(e,y)$} (8,0) -- (6,3.464) -- (4.25,0.433)--(4.5,0);

        \filldraw[red, opacity=0.3] (2,3.464) -- node[black, opacity=1, xshift=0.2cm,yshift=1.3cm] {\Large $\tau(e,z)$} (5.5,3.464) -- (5.75,3.897) -- (4,6.928) -- (2,3.464);
    \end{tikzpicture}
    \caption{Illustration of $\tau(e,\cdot)$ for $\tilde{e} = (x,y,z)$.}
    \label{fig: L1 tau diagram}
\end{figure}

To motivate this choice of $\tau$, we introduce the following terminology. First, for a given point $p = (p_1,p_2)\in U$, we refer to $p_1\in \had_m$ and $p_2\in\had_m$ as the first and second \textit{blocks} of $p$, respectively. For $A\subseteq U$, we say a Hadamard codeword $\textbf{h}_v\in \text{Had}_m^+$ \textit{features positively in $A$} if $\textbf{h}_v$ is a block (either the first or the second) of some point in $A$. Similarly, we say $\textbf{h}_v$ \textit{features negatively in $A$} if $\comp{\textbf{h}}_v$ is a block of some point in $A$.

Our choice of $\tau$, then, has the following \textit{exclusivity} property. For any $e=\{x,y,z\}\in E$, no codeword in $\text{Had}_m^+$ features both positively and negatively in $\tau(e,x)$: the codeword $\textbf{h}_x$ features positively but not negatively, while the codewords $\textbf{h}_y$ and $\textbf{h}_z$ feature negatively but not positively. From just this property, we can conclude that $\textsf{diam}(\tau(e,x))\leq m/2 + m/2 = m$, since no two points in $\tau(e,x)$ may contain complementary first blocks or complementary second blocks. Analogous statements hold for $\tau(e,y)$ and $\tau(e,z)$.

\begin{lemma} \label{lem: L1 cloud}
    Let $\mc{H}$ be a $3$-uniform hypergraph. Let $U,\rho,\tau$ be as in Definition \ref{def: L1 cloud}. Then, $(U,\rho,\tau)$ is an efficiently computable $3/2$-cloud system of  $\mc{H}$ with short distance $\beta = m$.
\end{lemma}
\begin{proof}
    It is clear that $\rho$ and $\tau$ can each be evaluated in polynomial time. Moreover, since the Hadamard code can be computed in polynomial time, the set $U$ can also be computed in polynomial time.

    We will first show that $(U,\rho,\tau)$ satisfies the proximity condition in Definition \ref{def: cloud}. Let $v,v'\in V$ be nonadjacent. Then, it suffices to show that for any $p,p'\in P_v\cup P_{v'}$, we have $\norm{p-p'}_1 \leq m$. Writing $p = (p_1,p_2)$ and $p' = (p_1',p_2')$ where $p_1,p_2,p_1',p_2'\in\had_m$, it suffices to show that $\norm{p_1-p_1'}_1, \norm{p_2-p_2'}_1\leq m/2$, or equivalently, that $p_1'\neq \comp{p_1}$ and $p_2'\neq\comp{p_2}$. It is enough, then, to prove that no codeword in $\text{Had}_m^+$ features both positively and negatively in $P_v\cup P_{v'}$.
    
    Observe that we have \[P_v\cup P_{v'} = \del{\bigcup_{\substack{e\in E \\ (e,v)\in I_{\mc{H}}}} \tau(e,v)}\cup \del{\bigcup_{\substack{e'\in E \\ (e',v')\in I_{\mc{H}}}} \tau(e',v')}\] By the exclusivity property of $\tau$, the only codewords featuring positively in $P_v\cup P_{v'}$ are $\textbf{h}_v$ and $\textbf{h}_{v'}$. Moreover, $\textbf{h}_v$ does not feature negatively in $P_v$, as it does not feature negatively in any $\tau(e,v)$. We claim that $\textbf{h}_v$ cannot feature negatively in $P_{v'}$ either. Otherwise, there must be some $e'\in E$ containing $v'$ such that $\textbf{h}_v$ features negatively in $\tau(e',v')$. But then $v\in e'$, contradicting the fact that $v,v'$ are nonadjacent. Therefore, $\textbf{h}_v$ does not feature negatively in $P_v\cup P_{v'}$. The same reasoning applies to $\textbf{h}_{v'}$, allowing us to conclude that $(U,\rho,\tau)$ satisfies the proximity condition. 
    
    Next, we will show that $(U,\rho,\tau)$ satisfies the spread condition in Definition \ref{def: cloud}. By symmetry, it suffices to show that any $3$-clustering of $P_{\{x,y,z\}}$ in which $(\textbf{h}_x,\textbf{h}_x)$ and $(\textbf{h}_y,\textbf{h}_y)$ are clustered together necessarily has diameter at least $3m/2$. Let $\{C_1,C_2,C_3\}$ be a $3$-clustering of $P_{\{x,y,z\}}$ and suppose for contradiction that $(\textbf{h}_x,\textbf{h}_x),(\textbf{h}_y,\textbf{h}_y)\in C_1$, where $\textsf{diam}(C_1) < 3m/2$. Then, no point in $C_1$ can contain the codewords $\comp{\textbf{h}}_x$ or $\comp{\textbf{h}}_y$, and thus \[C_1\subseteq \{(\textbf{h}_x,\textbf{h}_x), (\textbf{h}_y,\textbf{h}_y), (\textbf{h}_z,\textbf{h}_z), (\textbf{h}_x, \comp{\textbf{h}}_z), (\comp{\textbf{h}}_z,\comp{\textbf{h}}_z), (\comp{\textbf{h}}_z, \textbf{h}_y)\}.\] We divide into two cases. In the first case, $(\textbf{h}_z,\textbf{h}_z)\in C_1$, and so, to ensure $\textsf{diam}(C_1) < 3m/2$, it must be that \[C_1=\{(\textbf{h}_x,\textbf{h}_x), (\textbf{h}_y,\textbf{h}_y), (\textbf{h}_z,\textbf{h}_z)\}.\] Then, $(\comp{\textbf{h}}_x,\comp{\textbf{h}}_x),(\comp{\textbf{h}}_y,\comp{\textbf{h}}_y),(\comp{\textbf{h}}_z,\comp{\textbf{h}}_z)\in C_2\cup C_3$. Two of these three points must be in the same cluster; without loss of generality, say $(\comp{\textbf{h}}_x,\comp{\textbf{h}}_x),(\comp{\textbf{h}}_y,\comp{\textbf{h}}_y)\in C_2$. Then, in order to ensure that $\textsf{diam}(C_2) < 3m/2$, we must have $(\comp{\textbf{h}}_y, \textbf{h}_x), (\textbf{h}_y,\comp{\textbf{h}}_x)\in C_3$. But then $\textsf{diam}(C_3) \geq \norm{(\comp{\textbf{h}}_y, \textbf{h}_x) -(\textbf{h}_y,\comp{\textbf{h}}_x)}_1 =2m$.

    In the second case, $(\textbf{h}_z,\textbf{h}_z)\notin C_1$, so we have \[C_1\subseteq \{(\textbf{h}_x,\textbf{h}_x), (\textbf{h}_y,\textbf{h}_y), (\textbf{h}_x, \comp{\textbf{h}}_z), (\comp{\textbf{h}}_z,\comp{\textbf{h}}_z), (\comp{\textbf{h}}_z, \textbf{h}_y)\}.\]
    We will now show that either $C_2$ or $C_3$ necessarily has diameter $\geq 3m/2$. Without loss of generality, $(\textbf{h}_z,\textbf{h}_z)\in C_2$, which in turn implies $(\comp{\textbf{h}}_z,\comp{\textbf{h}}_x), (\comp{\textbf{h}}_y,\comp{\textbf{h}}_z)\in C_3$. But then $(\comp{\textbf{h}}_y, \textbf{h}_x), (\textbf{h}_y, \comp{\textbf{h}}_x)\in C_2$, a contradiction since $\textsf{diam}(C_2)\geq \norm{(\comp{\textbf{h}}_y, \textbf{h}_x)-(\textbf{h}_y, \comp{\textbf{h}}_x)}_1 = 2m$. We conclude that any $3$-clustering of $P_{\{x,y,z\}}$ with diameter $< 3m/2$ has $(\textbf{h}_x,\textbf{h}_x),(\textbf{h}_y,\textbf{h}_y),(\textbf{h}_z,\textbf{h}_z)$ in different clusters, as desired.
\end{proof}

\noindent This $3/2$-cloud system combined with Theorem \ref{thm: main reduction} gives us our first main result.

\begin{theorem}[\kdiam in the $\ell_1$-metric] 
\label{thm: L1 3/2}
    Let $k \ge 3$ be constant. Given $m\in \N$ and a pointset $P\subset (\R^m,\ell_1)$, it is \np-hard to distinguish between the following two cases: \begin{itemize}
        \item \textbf{\textup{Completeness:}} There exists a $k$-clustering of $P$ with diameter at most $1$.
        \item \textbf{\textup{Soundness:}} Any $k$-clustering of $P$ has diameter at least $3/2$.
    \end{itemize}
\end{theorem}
\begin{proof}
    By Lemma \ref{lem: L1 cloud}, for any $3$-uniform hypergraph $\mc{H} = (V,E)$, there is an efficiently computable $3/2$-cloud system $(U,\rho,\tau)$ of $\mc{H}$ where $U\subset (\R^m, \ell_1)$ and $m$ is a 
    power of $2$. By scaling all points, we can assume that the short distance of the system is $1$. Applying Theorem \ref{thm: main reduction} and hardness of panchromatic $k$-coloring on $k$-uniform hypergraphs, it is \np-hard to distinguish between the two cases for $k=3$. The statement follows by Remark \ref{rem: larger k is harder}.
\end{proof}

\noindent Theorem \ref{thm: L1 3/2} also immediately implies an inapproximability result in the $\ell_2$-metric.

\begin{corollary}[\kdiam in the $\ell_2$-metric] \label{cor: L2 sqrt(3/2)-hard}
    Let $k \ge 3$ be constant. Given $m\in \N$ and a pointset $P\subset (\R^m,\ell_2)$, it is \np-hard to distinguish between the following two cases: \begin{itemize}
        \item \textbf{\textup{Completeness:}} There exists a $k$-clustering of $P$ with diameter at most $1$.
        \item \textbf{\textup{Soundness:}} Any $k$-clustering of $P$ has diameter at least $\sqrt{3/2}$.
    \end{itemize}
\end{corollary}

\begin{proof} 
    Since the distance between any two points in $\{0,1\}^m$ 
    in the $\ell_1$-metric is the square of the distance between these points in the $\ell_2$-metric, any $r$-cloud system $(U,\rho,\tau)$ where $U\subset(\{0,1\}^m, \ell_1)$ is a $\sqrt{r}$-cloud system when $U$ is treated as a subset of $(\{0,1\}^m, \ell_2)$. Applying Theorem \ref{thm: main reduction} as above, we obtain the desired result.
\end{proof}

\section{Hardness of Approximation in the $\ell_2$-metric}\label{sec: L2}

In this section, we improve upon Corollary \ref{cor: L2 sqrt(3/2)-hard} by treating the $\ell_2$-metric directly. Indeed, we demonstrate the existence of a $1.304$-cloud system of any 3-uniform hypergraph in the $\ell_2$-metric, hence proving hardness of approximation \kdiam within a factor of $1.304$ in this setting.

We first establish the following notation for integer partitions.
    \begin{definition} \label{def: partition}
        For $\kappa,t\in\N$, let $\text{part}(\kappa, t) = \{(\alpha_1,\ldots,\alpha_t)\in(\N\cup\{0\})^t : \sum_{i=1}^t \alpha_i= \kappa\}.$
    \end{definition}

Now, we define a net of points on particular regions of the surface of a unit sphere.
\begin{definition} \label{def: L2 cloud}
    Let $\mc{H} = (V,E)$ be a 3-uniform hypergraph, and identify $V=[m]$ so that each vertex $v\in V$ corresponds to a distinct standard basis vector $\textbf{e}_v\in\R^m$. For $\kappa\in \N$, define

    \[U_\kappa = \cbr{\frac{\sum_{v \in [m]} \alpha_v \textbf{e}_v}{|\sum_{v \in [m]} \alpha_v \textbf{e}_v|} : (|\alpha_1|,\ldots,|\alpha_m|) \in \text{part}(\kappa, m)}. \] 
    Geometrically, $U_\kappa$ is a net of rational points on the surface of the unit sphere in $\R^m$ centered at $\textbf{0}$. We define an injective map $\rho : V\to U_\kappa$ given by $\rho(v) = \textbf{e}_v$. Next, we define $\tau_\kappa : I_\mc{H}\to \mc{P}(U)$ as follows. For every edge $e = \{x,y,z\}$, choose an arbitrary orientation $\tilde{e} = (x,y,z)$. Then, let $\tau_\kappa$ be given by: \begin{align*}
        \tau_\kappa(e,x) &= \cbr{\frac{\alpha_x \textbf{e}_x - \alpha_y\textbf{e}_y -\alpha_z\textbf{e}_z}{\abs{\alpha_x \textbf{e}_x - \alpha_y\textbf{e}_y -\alpha_z\textbf{e}_z}} :\; (\alpha_x,\alpha_y,\alpha_z)\in\text{part}(\kappa, 3)} \setminus \cbr{-\textbf{e}_y}\\
        \tau_\kappa(e,y) &= \cbr{\frac{-\alpha_x \textbf{e}_x +\alpha_y\textbf{e}_y -\alpha_z\textbf{e}_z}{\abs{-\alpha_x \textbf{e}_x +\alpha_y\textbf{e}_y -\alpha_z\textbf{e}_z}} :\;(\alpha_x,\alpha_y,\alpha_z)\in\text{part}(\kappa, 3)} \setminus \cbr{-\textbf{e}_z}\\
        \tau_\kappa(e,z) &= \cbr{\frac{-\alpha_x \textbf{e}_x - \alpha_y\textbf{e}_y +\alpha_z\textbf{e}_z}{|{-\alpha_x} \textbf{e}_x - \alpha_y\textbf{e}_y +\alpha_z\textbf{e}_z|} : (\alpha_x,\alpha_y,\alpha_z)\in\text{part}(\kappa, 3)} \setminus \cbr{-\textbf{e}_x}\\
    \end{align*}
    This definition of $\tau_\kappa$ is illustrated in Figure \ref{fig: L2 tau diagram}.
\end{definition}

\begin{figure}[!h]
    \centering
    
    \begin{tikzpicture}[scale=0.75]

        \shade[ball color = gray!40, opacity = 0.2] (0,0) circle (4);
        \draw (0,0) circle (4);
        \draw (-4,0) arc (180:360:4 and 0.6);
        \draw[dashed] (-4,0) arc (180:360:4 and -0.6);
        \draw (0,4) arc (90:270:0.5 and 4);

        \fill[cyan, opacity = 0.3] (0,4) arc (90:188.75:0.5 and 4) -- (-4,0) -- (0,4) -- cycle;
        \fill[cyan, opacity = 0.3] (-4,0) arc (180:98:4 and -0.6);
        \fill[cyan, opacity = 0.3] (-4,0) arc (180:90:4 and 4);
        \fill[green, opacity=0.3] (0,-4) arc (270:360:4 and 4);
        \fill[green, opacity=0.3] (4,0) arc (180:277:-4 and 0.6) -- (0,-4) -- (4,0) -- cycle;
        \fill[green, opacity=0.3] (0,-4) arc (270:360:-0.5 and 3.5) -- (0,-4) -- cycle;
        \fill[red, opacity=0.2] (4,0) arc (0:90:4 and 4) -- (4,0) -- cycle;
        \draw[dashed] (0,4) arc (90:171.5:-0.5 and 4);
        \fill[red, opacity=0.2] (0,4) arc (90:171.5:-0.5 and 4) arc (90:178.5:-3.5 and 0.6) -- (0,4) -- cycle;
    
        \filldraw(-3,-0.4) circle(1.5pt) node[xshift=0.6cm,yshift=1.4cm]{\Large $\tau_2(e,x)$};
        \filldraw(3,0.4) circle(1.5pt);
        \filldraw(3,-0.4) circle(1.5pt);
        \filldraw(-2.828,2.828) circle(1.5pt);
        \filldraw(2.828,-2.828) circle(1.5pt) node[xshift=-1cm,yshift=0.6cm]{\Large$\tau_2(e,y)$};
        \filldraw(2.828,2.828) circle(1.5pt);
        \filldraw(0.375,2.6) circle(1.5pt) node[xshift=1.1cm,yshift=-0.4cm]{\Large$\tau_2(e,z)$};
        \filldraw(-0.375,-2.6) circle(1.5pt);
        \filldraw(-0.375,2.6) circle(1.5pt);
        
        \filldraw[fill=green](4,0) node[xshift=0.5cm]{\small$-\textbf{e}_x$} circle(2pt);
        \filldraw[fill=red, opacity=1](0,4) node[yshift=0.3cm]{\small$-\textbf{e}_y$} circle(2pt);
        \filldraw[fill=cyan](-0.49,-0.6) node[xshift=-0.4cm, yshift=-0.3cm]{\small$-\textbf{e}_z$} circle(2pt);
        \filldraw(-4,0) node[xshift=-0.4cm]{\small$\textbf{e}_x$} circle(2pt);
        \filldraw(0,-4) node[yshift=-0.4cm]{\small$\textbf{e}_y$} circle(2pt);
        \filldraw(0.49,0.6) node[xshift=-0.1cm, yshift=-0.3cm]{\small$\textbf{e}_z$} circle(2pt); 
        
    \end{tikzpicture}
    \caption{Illustration of $\tau_\kappa(e, \cdot)$ for $\tilde{e} = (x,y,z)$ and $\kappa=2$. Here, each set $\tau_2(e,\cdot)$ contains $5$ points on the surface of a unit sphere. Note the similarity between $\tau_2$ and the construction in the $\ell_1$-metric, as shown in Figure \ref{fig: L1 tau diagram}.} 
    \label{fig: L2 tau diagram}
\end{figure}

This choice of $\tau_\kappa$ is motivated similarly to the one in Section \ref{sec: L1}. For $A\subseteq U$, we say that a coordinate $v\in [m]$ \textit{features positively in $A$} if $p_v > 0$ for some $p\in A$. Similarly, $v$ \textit{features negatively in $A$} if $p_v < 0$ for some $p\in A$.

Our choice of $\tau_\kappa$, as in Section \ref{sec: L1}, has the following \textit{exclusivity} property. For any $e=\{x,y,z\}\in E$, no coordinate in $[m]$ features both positively and negatively in $\tau_\kappa(e,x)$: the coordinate $x$ features positively but not negatively, while the coordinates $y$ and $z$ feature negatively but not positively. Analogous statements hold for $\tau_\kappa(e,y)$ and $\tau_\kappa(e,z)$.


\begin{lemma} \label{lem: L2 cloud}
    Let $U_\kappa,\rho,\tau_\kappa$ be as in Definition \ref{def: L2 cloud}. Then, for a suitable choice of constant $\kappa$, $(U_\kappa,\rho,\tau_\kappa)$ is an efficiently computable $1.304$-cloud system with short distance $\beta = \sqrt{2}$.
\end{lemma}
\begin{proof}
    For constant $\kappa\in\N$, it is clear that $\rho$, $\tau_\kappa$, and $U_\kappa$ can each be computed in polynomial time. 

    We will first show that $(U_\kappa,\rho,\tau_\kappa)$ satisfies the proximity condition in Definition \ref{def: cloud} for any $\kappa\in\N$. Let $v,v'\in V$ be nonadjacent. Then, it suffices to show that for any $p,p'\in P_v\cup P_{v'}$, we have $\norm{p-p'}_2 \leq \sqrt{2}$. Since $p$ and $p'$ are both unit vectors, this is equivalent to proving $\ip{p}{p'}\geq 0$, as \[\norm{p-p'}_2 = \sqrt{\ip{p-p'}{p-p'}} = \sqrt{2} - \sqrt{2}\cdot\ip{p}{p'}.\] It is enough, then, to show that no coordinate features both positively and negatively in $P_v\cup P_{v'}$, as this implies that $p_i{p_i}'\geq 0$ for each $i\in [m]$ and thus $\ip{p}{p'}\geq 0$.
        
    Observe that we have \[P_v\cup P_{v'} = \del{\bigcup_{\substack{e\in E \\ (e,v)\in I_{\mc{H}}}} \tau_\kappa(e,v)}\cup \del{\bigcup_{\substack{e'\in E \\ (e',v')\in I_{\mc{H}}}} \tau_\kappa(e',v')}\] By the exclusivity property of $\tau_\kappa$, the only coordinates featuring positively in $P_v\cup P_{v'}$ are $v$ and $v'$. Moreover, $v$ does not feature negatively in $P_v$, as it does not feature negatively in any $\tau_\kappa(e,v)$. We claim that $v$ cannot feature negatively in $P_{v'}$ either. Otherwise, there must be some $e'\in E$ containing $v'$ such that $v$ features negatively in $\tau_\kappa(e',v')$. But then $v\in e'$, contradicting the fact that $v,v'$ are nonadjacent. Therefore, $v$ does not feature negatively in $P_v\cup P_{v'}$. The same reasoning applies to $v'$, allowing us to conclude that $(U_\kappa,\rho,\tau_\kappa)$ satisfies the proximity condition for any $\kappa \in \N$. 

    We claim that $(U_\kappa,\rho,\tau_\kappa)$ satisfies the spread condition in Definition \ref{def: cloud} for $\kappa = 12$. This can be verified using a computer by confirming that all $3$-clusterings of $P_{\{x,y,z\}}\subset U_{12}$ in which $\textbf{e}_x$ and $\textbf{e}_y$ are clustered together have diameter at least $1.304 \cdot \sqrt{2}$. The code used to prove the second part of Lemma \ref{lem: L2 cloud} can be found at \href{https://github.com/cea4608937/hardness-of-diameter}{https://github.com/cea4608937/hardness-of-diameter}.
\end{proof}

Applying Theorem \ref{thm: main reduction} to this cloud system gives us our main result in the $\ell_2$-metric.

\begin{theorem}[\kdiam in the $\ell_2$-metric] \label{thm: L2 1.304}
    Let $k \ge 3$ be constant. Given $m\in \N$ and a pointset $P\subset (\R^m,\ell_2)$, it is \np-hard to distinguish between the following two cases: \begin{itemize}
        \item \textbf{\textup{Completeness:}} There exists a $k$-clustering of $P$ with diameter at most $1$.
        \item \textbf{\textup{Soundness:}} Any $k$-clustering of $P$ has diameter at least $1.304$.
    \end{itemize}
\end{theorem}
\begin{proof}
    By Lemma \ref{lem: L2 cloud}, for any $3$-uniform hypergraph $\mc{H} = (V,E)$, there is an efficiently computable $3/2$-cloud system $(U,\rho,\tau)$ of $\mc{H}$ where $U\subset (\R^m, \ell_1)$ for $m=|V|$. By scaling all points, we can assume that the short distance of the system is $1$. Applying Theorem \ref{thm: main reduction} and hardness of panchromatic $k$-coloring on $k$-uniform hypergraphs, it is \np-hard to distinguish between the two cases for $k=3$. The statement follows by Remark \ref{rem: larger k is harder}.
\end{proof}

\begin{remark}
    Fix an edge $e = \{x,y,z\}$, and consider the set $P_{\{x,y,z\}}\subset U_\kappa$. Note that as $\kappa$ increases, $P_{\{x,y,z\}}$ forms denser nets, allowing us to potentially satisfy the spread condition for larger values of $r$, thereby obtain larger hardness factors than $1.304$. However, there is not much room for improvement with this technique. Namely, the hardness factor we obtain will never surpass $\sqrt{1+\sqrt{2}/2} \approx 1.307$ because for this particular construction, the spread property does not hold when $r\geq \sqrt{1+\sqrt{2}/2}$. In particular, we define the following $3$-clustering of $P_{\cbr{x,y,z}}\subset U_\kappa$. Denote the clusters $C_x, C_y, C_z$ and for $p\in P_{\cbr{x,y,z}}$, let $p\in C_{v'}$ if \[v' = \argmin_{v\in\cbr{x,y,z}} p_v,\] deciding ties arbitrarily with a single exception: we force that $\textbf{e}_x, \textbf{e}_y\in C_z$, which is possible since $(\textbf{e}_x)_z = (\textbf{e}_y)_z = 0$. It can be checked that for any $\kappa\in\N$, the clustering $\{C_x,C_y,C_z\}$ has diameter at most $\sqrt{2+\sqrt{2}} = \sqrt{2} \cdot \sqrt{1+\sqrt{2}/2}$. Since $\rho(x)$ and $\rho(y)$ clustered together, this confirms that the spread property does not hold. 
    \end{remark}

\section{Barriers to Proving Hardness of \kdiam} \label{sec: barriers}

In this section, we present two barriers to showing improved hardness of approximation for the \kdiam problem.

\subsection{Intermediate Distance Barrier}
In this section, we show a barrier to our methods of proving hardness of \kdiam in the $\ell_1$-metric. We show that proving hardness of approximation to a factor $r > 5/3$ requires constructing a pointset without large gaps in the set of pairwise distances.

\begin{definition}
    Let $P\subset (X,\textsf{dist})$ be a set of points. For $a,b\in\R$, we say that $P$ has an $(a,b)$-gap if there are no $p,p'\in P$ such that $\textsf{dist}(p,p')\in (a,b)$.
\end{definition}

\begin{definition}
    Let $G$ and $H$ be graphs. We say that $G$ is $H$-free if no induced subgraph of $G$ is isomorphic to $H$.
\end{definition}

It turns out that there is a connection between \kdiam on pointsets with gaps and coloring $H$-free graphs.

\begin{theorem} \label{thm: L1 barrier}
    Let $H$ be a graph such that $H$-free $k$-coloring is in \p and H is not $r$-embeddable. Then, there is a polynomial time algorithm that takes as input pointsets $P$ with a $(\beta, r\beta)$-gap and determines if the optimal $3$-clustering of $P$ has diameter at most $\beta$.
\end{theorem}

\begin{proof}
    Let $\mathcal{A}$ be a polynomial time algorithm which, given an $H$-free graph, determines if a $k$-coloring exists. Given a pointset $P$ with an $(\beta,r\beta)$-gap, construct a graph $G = (V,E)$ where $V = P$ and $(p,p')\in E$ if $\textsf{dist}(p,p') > \beta$. 
    We observe that the graph $G$ must be $H$-free, or else the points corresponding to $H \subset G$ would give an $r$-embedding of $H$. By running $\mathcal{A}$ on $G$, we determine if a $k$-coloring of $G$ exists.  A coloring of $G$ directly corresponds to a $k$-clustering of $P$ with diameter at most $\beta$, so this algorithm determines if $P$ has a $k$-clustering with diameter at most $\beta$, as desired.
\end{proof}

In order to apply Theorem \ref{thm: L1 barrier}, we must identify a graph $H$ for which $H$-free $k$-coloring can be solved in polynomial time.

\begin{lemma}\cite{Bonomo2018}\label{lem: P7 free easy}
    Let $P_7$ denote the path graph on $7$ vertices. There exists a polynomial time algorithm to $3$-color $P_7$-free graphs.
\end{lemma}

\begin{corollary} \label{cor: L0 barrier}
    For any $r > 5/3$, there is a polynomial time algorithm that takes as input pointsets $P$ in the $\ell_1$-metric with a $(\beta, r\beta)$-gap and determines if the optimal $3$-clustering of $P$ has diameter at most $\beta$.
\end{corollary}
\begin{proof}
    Using the method in Appendix~\ref{sec: appendix A}, we verify that the graph $P_7$ is not $r$-embeddable in the $\ell_1$-metric for any $r > 5/3$. Applying Lemma \ref{lem: P7 free easy} and Theorem \ref{thm: L1 barrier} gives the result.
\end{proof}

\begin{remark}
    Corollary \ref{cor: L0 barrier} can be seen as a barrier to various methods of proving hardness of approximation results in the $\ell_1$-metric. First, consider the initial $r$-embedding approach for proving hardness of \kdiam outlined in Section \ref{sec: introduction}. In this approach, we construct an $r$-embedding $\varphi$ (with short distance $\beta$) of a hard-to-color graph $G$, arguing that it is \np-hard to determine if the optimal $3$-clustering of $\varphi(G)$ has diameter at most $\beta$. Importantly, by definition of $r$-embedding, the pointset $\varphi(G)$ has a $(\beta, r\beta)$-gap.
    
    Next, recall the $3/2$-cloud system constructed in Section \ref{sec: L1} and the associated pointset $P$. By the unscaled version of Theorem \ref{thm: L1 3/2}, it is \np-hard to determine if the optimal $3$-clustering of $P$ has diameter at most $m$. Moreover, every pairwise distance in $\had_m\times \had_m$ lies in $\{0,m/2,m,3m/2,2m\}$, implying that $P$ has an $(m, 3m/2)$-gap.

    Corollary \ref{cor: L0 barrier} implies that any pointset on which \kdiam cannot be approximated to a $(5/3+\varepsilon)$ factor must contain intermediate distances, which is not the case in the above two approaches. Note that the approach in Section \ref{sec: L2} does have such intermediate distances, so forming nets of points gives one possible method of circumventing this barrier.
\end{remark}

\subsection{Large Odd Girth Barrier}
In this section, we show a better than $\sqrt{2}$-approximation algorithm for \tdiam, restricted to pointsets that are not contained within any ball of diameter $\Delta\cdot \sqrt{2}$, where $\Delta$ denotes the optimal $3$-clustering diameter. We then give a graph theoretic implication of this restriction related to the existence of short odd cycles. 

We start by formally describing the $(\sqrt{2}+\varepsilon)$-approximation algorithm for \kdiam mentioned in Section \ref{sec: introduction}. The natural setting of this algorithm is the closely related $k$-center clustering problem, which we define below.

\begin{definition}[$k$-center clustering problem]
    Let $(X, \textsf{dist})$ be a metric space and $k \in \N$. Given as input a finite set $P \subset X$, find a $k$-clustering $\{A_1,\ldots,A_k\}\subset P$ and a collection of ``centers'' $\{a_1,\ldots,a_k\}$ that minimizes
    \[ \max_{i\in [k],\, x\in A_i} \textsf{dist}(x,a_i).\]
\end{definition}

\begin{theorem}[follows from~\cite{badoiu2002approximate}]\label{alg:const_sqrt_approx}
    Let $\varepsilon>0$ and $k$ a constant. In Euclidean space, $(\sqrt{2}+\varepsilon)$-approximate $k$-clustering is in \p. 
\end{theorem}
\begin{proof}
    From~\cite{badoiu2002approximate}, $(1+\varepsilon)$-approximation of Euclidean $k$-center is in \p. By Jung's theorem~\cite{danzer1963helly}, any set of points with diameter $\Delta$ is contained in a closed ball with radius $r$ satisfying
    \[ r \leq \Delta \sqrt{\frac{n}{2(n+1)}} \leq \frac{\Delta}{\sqrt{2}} \]
    Let $\Delta$ be the optimal \kdiam of the given set of points.
    Applying Jung's theorem, the optimal clustering for the $k$-center objective will have diameter
    \[ \Delta' = 2r \leq \Delta\sqrt{2}\]
    Therefore, $(1+\varepsilon)$-approximation of $k$-center immediately gives $(\sqrt{2}+\varepsilon)$-approximation of \kdiam.
\end{proof}

In the case of $k=3$, we can modify this algorithm to a better approximation for certain pointsets.

\begin{theorem} \label{thm: in a sphere barrier}
    Let $\varepsilon>0$ be a sufficiently small constant. Then, there exists a polynomial time $(\sqrt{2}-\varepsilon)$-approximation algorithm for \tdiam, when restricted to pointsets $P$ not contained in a bounding sphere with diameter $\Delta\cdot\left(\sqrt{2}+50\varepsilon^{1/8}\right)$, where $\Delta$ is the optimal $3$-diameter of $P$.
\end{theorem}
\noindent The proof of Theorem \ref{thm: in a sphere barrier} can be found in Appendix \ref{sec: appendix B}. Importantly, Theorem \ref{thm: in a sphere barrier} has implications for the graph $\Gamma_{k,r}$ introduced in Definition \ref{def: gamma}.

\begin{definition}
    Given a graph $G$, the odd girth of $G$ is the length of the smallest odd cycle. If there are no odd cycles, we say that the odd girth is infinite.
\end{definition}

\begin{theorem} \label{thm: odd girth barrier}
For any $g\in\N$, there exists an $\varepsilon>0$ such that $(\sqrt{2}-\varepsilon)$-approximate \tdiam is in \p, when restricted to pointsets $P$ such that $\Gamma_{k,\sqrt{2}-\varepsilon}(P)$ has odd girth at most $g$. 
\end{theorem}
\begin{proof}
    By scaling, we may assume without loss of generality that the optimal $3$-clustering diameter of $P$ is 1. We will prove that for small enough $\varepsilon$, there is no pointset $P$ contained in a ball of diameter $\sqrt{2}+\varepsilon$ such that the odd girth of $\Gamma_{k,\sqrt{2}-\varepsilon}(P)$ is at most $g$. Then, the statement follows from Theorem \ref{thm: in a sphere barrier}.

    So, suppose for contradiction that $P$ is a pointset contained in a ball of diameter $\sqrt{2}+\varepsilon$ and the odd girth of $\Gamma_{k,\sqrt{2}-\varepsilon}(P)$ is at most $g$. We can assume that this ball is centered at the origin. Since the optimal diameter is $1$, all edges in $\Gamma_{k,\sqrt{2}-\varepsilon}(P)$ correspond to distances of at least $\sqrt{2}-\varepsilon$. Let $u,u' \in \Gamma_{k,\sqrt{2}-\varepsilon}(P)$ be arbitrary adjacent vertices and let $p,p'\in P$ be their corresponding points. Note that $\norm{p-p'} \ge \sqrt{2}-\varepsilon$. Since both $p$ and $p'$ are contained in a sphere of diameter $\sqrt{2}+\varepsilon$ centered at the origin, we have that $\norm{p+p'}_2 \le \frac{1}{2g}$ for small enough $\varepsilon$. Thus, the distance between any two points in $P$ whose corresponding vertices share a neighbor must be at most $2 \cdot \frac{1}{2g} = \frac{1}{g}$ by the triangle inequality.

    By assumption, there must be an odd cycle $C$ in $\Gamma_{k,\sqrt{2}-\varepsilon}(P)$ of consisting of at most $g$ vertices. Let $v,v'\in C$ be adjacent vertices and $q,q'\in P$ be their corresponding points. Note that since the length of the cycle $C$ is odd, there is an even length path between $v$ and $v'$ in $\Gamma_{k,\sqrt{2}-\varepsilon}(P)$ that has length at most $g-1$. 
    Let $q=q_0, q_1, \dots, q_m =q' \in P$ be every other point along this path, with $m\leq \frac{g}{2}$. For each $i,\, 0\leq i\leq m-1$, the vertices corresponding to $q_i$ and $q_{i+1}$ share a neighbor, and thus \[\norm{q_i-q_{i+1}}\leq  \frac{1}{g}.\] By the triangle inequality, we have \[\norm{q-q'}_2\leq  m \cdot \frac{1}{g}\leq \frac{1}{2}.\] For small enough $\varepsilon>0$, we have that $1/2 < \sqrt{2}-\varepsilon$, contradicting the fact that $q$ and $q'$ correspond to adjacent vertices.
\end{proof}

\begin{remark}
    We note that Theorems \ref{thm: in a sphere barrier} and \ref{thm: odd girth barrier} hold only with respect to \tdiam, and not \kdiam for $k\geq 4$. This is in contrast to our inapproximability result from Theorem \ref{thm: L2 1.304}, which holds for all $k\geq 3$.
\end{remark}

\subsection*{Acknowledgements}
We thank Vincent Cohen-Addad and Euiwoong Lee for pointing us to \cite{badoiu2002approximate} and it's implications to Euclidean \kdiam. 

This work was carried out while Kyrylo, Ashwin, and Stepan were participating in the  DIMACS REU 2023.
Kyrylo was supported by the CoSP, an European Union's Horizon 2020 programme, grant agreement No. 823748.
Karthik was supported   by the National Science Foundation under Grant CCF-2313372 and a grant from the Simons Foundation, Grant Number 825876, Awardee Thu D. Nguyen.
Ashwin and Stepan were supported by the NSF grant CNS-2150186.

\bibliographystyle{alpha}
\bibliography{references}

\appendix

\section{Testing $r$-embeddability with a Linear Program} \label{sec: appendix A}

It turns out finding the largest ratio $r$ for which a given graph $H$ is $r$-embeddable into $(\{0,1\}^m, \ell_1)$ can be framed as the solution to a linear program. To see this, let $H=(V,E)$ be a graph on $n$ vertices labeled $v_1, v_2, \ldots, v_n$, and let $\varphi$ be an $r$-embedding of $H$ into the space $\{0,1\}^m$, with short distance $\beta$. For each $w\in\{0,1\}^n$, associate a nonnegative integer $x_w$ indicating the number of indices $i\in[m]$ such that $(\varphi(v_j))_i = w_j$ holds for all $j\in [n]$. Then, distances in our embedded space can be written as the following linear combination: $$\norm{\varphi(v_a)-\varphi(v_b)}_0 = \sum_{\substack{w\in\{0,1\}^n \\ w_a\neq w_b}} x_w.$$ Then, for a given short distance parameter $\beta$, finding the largest possible embeddability ratio $r$ is equivalent to solving the following optimization problem over nonnegative integer variables $\{x_w : w\in\{0,1\}^n\}$ and a rational variable $r$:
\begin{equation*}
\begin{array}{ll@{}ll}
\text{maximize} & r &\\
\text{subject to} &
\begin{cases}\displaystyle\sum_{\substack{w\in\{0,1\}^n \\ w_a\neq w_b}} x_w \geq r\beta & \text{if } (v_a,v_b)\in E\\
\displaystyle\sum_{\substack{w\in\{0,1\}^n \\ w_a\neq w_b}} x_w \leq \beta & \text{if } (v_a,v_b)\notin E
\end{cases}
\end{array}
\end{equation*}
We can remove dependence on $s$ by simply setting $\beta=1$ and allowing each $x_w$ to take any nonnegative rational value, obtaining the following linear program: \begin{equation*}
\begin{array}{ll@{}ll}
\text{maximize} & r &\\
\text{subject to} &
\begin{cases}\displaystyle\sum_{\substack{w\in\{0,1\}^n \\ w_a\neq w_b}} x_w \geq r & \text{if } (v_a,v_b)\in E\\
\displaystyle\sum_{\substack{w\in\{0,1\}^n \\ w_a\neq w_b}} x_w \leq 1 & \text{if } (v_a,v_b)\notin E
\end{cases}
\end{array}
\end{equation*}

Moreover, since each of our constraints have integer coefficients, the solution vector to this linear program will be rational. Thus, we can scale each $x_w$ to be an integer, thus obtaining an embedding into $\{0,1\}^m$ that achieves the optimal ratio $r$, where $m 
 =  \displaystyle\sum_{\substack{w\in\{0,1\}^n}} x_w \in \N$.

In particular, for a given graph $H$ and parameter $r>1$, this gives us a computational way to check whether or not $H$ is $r$-embeddable in the Hamming metric. That said, the number of constraints in the linear program is $\Omega(2^n)$, meaning this method is only tractable when $n =|H|$ is relatively small.

\section{Proof of Theorem \ref{thm: in a sphere barrier}} \label{sec: appendix B}

\begin{proof}

Let $\varepsilon'=45\varepsilon^{1/4}$ and fix a pointset $P$ not contained in any bounding sphere with diameter $ \sqrt{2}+\varepsilon'$. By scaling, we can assume that $\Delta=\frac{\sqrt{2}}{\sqrt{2}+\varepsilon}$ without loss of generality. Let $A\cup B\cup C$ be an optimal $3$-clustering of $P$. 

    We run the algorithm from Theorem \ref{alg:const_sqrt_approx} on the pointset $P$, with approximation factor $\sqrt{2}+\varepsilon$, hence obtaining a $3$-clustering of $P$ with diameter at most $\Delta\cdot(\sqrt{2}+\varepsilon) = \sqrt{2}$. Theorem 2.7 of \cite{badoiu2002approximate} gives that all three clusters can be written as the intersection of balls in $\R^m$ with the pointset $P$, with points contained in multiple balls assigned arbitrarily. Moreover, we desire a stronger property: that the optimal clusters $A,B,C$ are each contained in one of these balls. Slightly extending their algorithm by returning all $3$-clusterings enumerated by the simulation of the ``guessing oracle'', we can ensure that one of the returned $3$-clusterings achieves a $(\sqrt{2}+\varepsilon)$-approximation factor as well as the desired property. It suffices to run the following procedure over all of these $3$-clusterings. On at least one iteration, the desired approximation will hold.

    Let $N_A,N_B,N_C$ be balls of diameter $\sqrt{2}$ that contain $A,B,C$, respectively. Let $a,b,c$ be their respective centers. 
    We now split into two cases which, by a relabeling argument, are exhaustive. \begin{itemize}
        \item \textbf{Case I:} $N_A\cap N_B,\,N_A\cap N_C,\,N_B\cap N_C$ each have diameter $\leq \sqrt{2}-\varepsilon$. In this case, we have the following algorithm:
        \begin{enumerate}
            \item For each point $p \in P$, initialize a set $S_p$ as follows: 
            \[S_p = \{X \in \{A,B,C\}: p \in N_X\}\]
            as the set of balls that contain $p$. 
            \item For every pair of points $p, q \in P$ such that $\norm{p-q}_2 >\Delta$, if $S_p$ contains only one element, and $S_q$ contains that element, remove it from $S_q$. Repeat this until no such removals are possible.
            \item Let 
            \begin{align*}
                A' &:= \{p \in P : A \in S_p, B \not \in S_p \} \\
                B' &:= \{p \in P : B \in S_p, C \not \in S_p \} \\
                C' &:= P \setminus (A' \cup B').
            \end{align*}
            Output the $3$-clustering $\{A',B',C'\}$.
        \end{enumerate}
        We can show that the following property remains true throughout the iterations of step 2 of the algorithm: if $p$ is in some optimal cluster $X \in \{A,B,C\}$, then $X \in S_p$.  It is true in the initialization of $S_p$ in step 1 of the algorithm. We only remove $X$ from $S_p$ if there exists some $q$ such that $\norm{p -q}_2 > \Delta$, and $q \in X$. So, $p$ cannot be in $X$, because the optimal clustering $\{A,B,C\}$ has diameter $\Delta$. Thus, if $p\in X$, then $X \in S_p$. \\

        Let $p,q \in A'$. Note that $A \in S_p$ and $A\in S_q$ by definition. If $S_p = \{A\}$ or $S_q = \{A\}$, then $\norm{p-q}_2 \le \Delta$, or $A$ would have been removed. In all other possible cases, $C \in S_p$ and $C \in S_q$. So, $p, q \in N_A \cap N_C$. Thus, $\norm{p-q}_2 \le \sqrt{2}-\varepsilon$.\\

        Similarly, if $p,q \in B'$, then either $\norm{p-q}_2 \le \Delta$ or $p, q \in N_A \cap N_B$. If $p,q \in C'$, then either $\norm{p-q}_2 \le \Delta \le \sqrt{2}- \varepsilon$ or $p, q \in N_C \cap N_B$. In all cases, $\norm{p-q}_2 \le \sqrt{2}-\varepsilon$, so our $3$-clustering has diameter at most $\sqrt{2}-\varepsilon$.

        \item \textbf{Case II:} $N_A\cap N_B$ has diameter $> \sqrt{2}-\varepsilon$.
        We claim that, in this case, the diameter of $C\cap (A\cup B)$ is at most $\sqrt{2}-\varepsilon$. Given that this is true consider the following algorithm:
        \begin{enumerate}
            \item For each point $p \in P$, initialize
            \[S_p = \{X \in \{A,B,C\}: p \in N_X\}\]
            as the set of balls that contain $p$, as in Case I. 
            \item For every pair of points $p, q \in P$ such that $\norm{p-q}_2 > \Delta$, if $S_p$ contains only one element, and $S_q$ contains that element, remove it from $S_q$. Repeat this until no such removals are possible, as in Case I.
            \item Let 
            \begin{align*}
                C' &:= \{p \in P : C \in S_p \}.
            \end{align*}
            Find an optimal $2$-clustering of $P \setminus C'$, and call the output clusters $A'$ and $B'$. Output the clustering $A',B',C'$.
        \end{enumerate}
        
        Note that for any point $p \in P \setminus C',$ we have that $C \not \in S_p$. Thus, $P \setminus C' \subseteq A \cup B$, so there exists a 2-clustering of $P \setminus C'$ of diameter at most $\Delta$. So, $A'$ and $B'$ have diameter at most $\Delta$.

        Now, for any $p,q \in C'$, if $S_p = \{C\}$ or $S_q = \{C\}$, then $\norm{p-q}_2 \le \Delta$, or $C$ would have been removed. So, if $\norm{p-q}_2 > \Delta$, then both $S_p$ and $S_q$ contain either $A$ or $B$. In this case, $p,q \in C \cap (A \cup B)$, so  $\norm{p-q}_2 \le \sqrt{2}-\varepsilon$. Hence, the clustering $\{A',B',C'\}$ has diameter at most $\sqrt{2}-\varepsilon$. \\
        
        It remains to verify our claim that $C\cap(A\cup B)$ has diameter at most $\sqrt{2}-\varepsilon$. First, observe that $N_A\cup N_B$ are contained within a ball $N_D$ whose center $d$ is the midpoint between $a$ and $b$ and whose radius is $\frac{\alpha+\sqrt{2}}{2}$, where $\alpha := \norm{a-b}_2$. Since $N_A\cap N_B$ has diameter at least $\sqrt{2}-\varepsilon$, by the Pythagorean Theorem we have: $$\alpha \leq 2\sqrt{\left(\frac{\sqrt{2}}{2}\right)^2-\left(\frac{\sqrt{2}-\varepsilon}{2}\right)^2} = \sqrt{2\varepsilon\sqrt{2}-\varepsilon^2} .$$ Now, we apply the fact that $P$ is not contained in a bounding sphere of diameter $\Delta\cdot(\sqrt{2}+\varepsilon')$. If $\varepsilon'\geq 2\varepsilon$, we observe that $$\Delta\cdot(\sqrt{2}+\varepsilon') = \sqrt{2}\cdot\frac{\sqrt{2}+\varepsilon'}{\sqrt{2}+\varepsilon} = \sqrt{2} + \sqrt{2}\cdot\frac{\varepsilon
        '-\varepsilon}{\sqrt{2}+\varepsilon}\geq \sqrt{2} + \sqrt{2}\cdot\frac{\varepsilon'/2}{\sqrt{2}} = \sqrt{2}+\frac{\varepsilon'}{2}.$$
        Hence, it is also true that $P$ is not contained in a bounding sphere of diameter $\sqrt{2}+\frac{\varepsilon'}{2}$. The same is true for $N_C\cup N_D$, since $P\subset N_C\cup N_D$, meaning: \begin{align*}
            \sqrt{2}+\frac{\varepsilon'}{2} &\leq 
            \text{diameter}(N_C\cup N_D) \\&\leq
            \norm{c-d}_2 + \text{radius}(N_C) + \text{radius}(N_D) \\&=
            \norm{c-d}_2 + \frac{\sqrt{2}}{2} + \frac{\alpha + \sqrt{2}}{2} \\&=
            \norm{c-d}_2 + \frac{\alpha}{2} + \sqrt{2}.
        \end{align*}
        Hence, $\norm{c-d}_2 \geq \varepsilon'-\frac{\alpha}{2}$. Finally, let $h$ denote the diameter of $N_C\cap N_D$;  since $A\cup B\subset N_D$ and $C\subseteq N_C$, it suffices to show that $h\leq \sqrt{2}-\varepsilon$. By the Pythagorean Theorem, we have: \begin{align*}
            \frac{h}{2} &\leq
            \sqrt{\text{radius}(N_D)^2 - \left(\frac{\norm{c-d}_2}{2}\right)^2} \\&\leq
            \frac{1}{2}\cdot \sqrt{(\alpha+\sqrt{2})^2 - \left(\frac{\varepsilon'-\alpha}{2}\right)^2} \\&=
            \frac{1}{2}\cdot\sqrt{\frac{3\alpha^2}{4} + 2\alpha\sqrt{2} + \frac{\alpha\varepsilon'}{2} + 2 - \frac{{\varepsilon'}^2}{4}}
        \end{align*} Hence, it suffices to show that: $$(\sqrt{2}-\varepsilon)^2 = 2 - 2\varepsilon\sqrt{2}+\varepsilon^2 \geq  \frac{3\alpha^2}{4} + \left(2\sqrt{2} + \frac{\varepsilon'}{2}\right)\cdot \alpha + 2 - \frac{{\varepsilon'}^2}{4}$$ From here, we provide a sequence of stronger inequalities, eventually showing that our setting of $\varepsilon' = 60\varepsilon^{1/4}$ is satisfies the above inequality. First, recall that $\alpha\leq \sqrt{2\varepsilon\sqrt{2}-\varepsilon^2}$. Thus, it is sufficient to have $$2 - 2\varepsilon\sqrt{2}+\varepsilon^2\geq \frac{6\varepsilon\sqrt{2}-3\varepsilon^2}{4} + \left(2\sqrt{2}+\frac{\varepsilon'}{2}\right)\sqrt{2\varepsilon\sqrt{2}-\varepsilon^2}  + 2 - \frac{{\varepsilon'}^2}{4}$$ Subtracting $(2+\frac{6\varepsilon\sqrt{2}-3\varepsilon^2}{4})$ from both sides, it is equivalent that $$\frac{7}{4}(\varepsilon^2-2\varepsilon\sqrt{2})\geq \left(2\sqrt{2}+\frac{\varepsilon
        '}{4}\right)\sqrt{2\varepsilon\sqrt{2}-\varepsilon^2}-\frac{{\varepsilon'}^2}{4}$$
        Applying the trivial bound $\varepsilon^2 \geq 0$, it is sufficient that $$-\frac{7\varepsilon}{\sqrt{2}} \geq \left(2\sqrt{2}+\frac{\varepsilon'}{2}\right)\sqrt{2\varepsilon\sqrt{2}}-\frac{{\varepsilon'}^2}{4}$$ For $\varepsilon'\leq 1$, we have $\left(2\sqrt{2}+\frac{\varepsilon'}{2}\right)\sqrt{2\sqrt{2}} < 6$, so it would be sufficient if $$-\frac{7\varepsilon}{\sqrt{2}} \geq 6\sqrt{\varepsilon}-\frac{{\varepsilon'}^2}{4}.$$ Rearranging, we have: $${\varepsilon'}^2 \geq 24\sqrt{\varepsilon} + \frac{28\varepsilon}{\sqrt{2}}.$$ Finally, for $\varepsilon\leq 1$ it is true that $\varepsilon\leq \sqrt{\varepsilon}$, and since $24+\frac{28}{\sqrt{2}} < 45$, it is good enough to have $${\varepsilon'}^2\geq 45\sqrt{\varepsilon}$$ This is satisfied when $\varepsilon' = 45\varepsilon^{1/4}.$ Note that our argument relied on $\varepsilon\leq 1$, $\varepsilon'\leq 1$ and $\varepsilon'\geq 2\varepsilon$ all holding. One can verify that for  $\varepsilon' = 45\varepsilon^{1/4}$, all of these bounds are true when $\varepsilon$ is sufficiently small.
    \end{itemize}

    \noindent In both cases, then, we produced an algorithm which outputs a $3$-clustering with diameter at most $\sqrt{2}-\varepsilon$. Since the optimal $3$-clustering has diameter $\Delta = \frac{\sqrt{2}}{\sqrt{2}+\varepsilon}$, we achieve an approximation ratio of: $$\frac{\sqrt{2}-\varepsilon}{\Delta} = \frac{(\sqrt{2}-\varepsilon)(\sqrt{2}+\varepsilon)}{\sqrt{2}} = \frac{2-\varepsilon^2}{\sqrt{2}} = \sqrt{2} - \frac{\varepsilon^2}{\sqrt{2}}.$$

    We have demonstrated a $(\sqrt{2}-\frac{\varepsilon^2}{\sqrt{2}})$-approximation algorithm for pointsets not contained in a bounding sphere of diameter $\Delta\cdot\left(\sqrt{2}+45\varepsilon^{1/4}\right)$. By a change of variables $\varepsilon\mapsto 2^{1/4}\sqrt{\varepsilon}$, this gives a $(\sqrt{2}-\varepsilon)$-approximation for pointsets not contained in a bounding sphere with diameter
    \[\Delta\cdot\left(\sqrt{2} + 45\left(2^{1/4}\sqrt{\varepsilon}\right)^{1/4}\right) = \Delta\cdot\left(\sqrt{2} + 45\cdot 2^{1/16} \varepsilon^{1/8}\right).\] Since $45\cdot 2^{1/16}<50$, this algorithm also holds for pointsets not contained in a bounding sphere with diameter $\Delta\cdot\left(\sqrt{2} +50\varepsilon^{1/8}\right)$, which is what we wanted to show.

\end{proof}

\end{document}